\documentclass[article,oneside,12pt,a4paper]{memoir}

\usepackage[utf8]{inputenc}
\usepackage[T1]{fontenc}
\usepackage[shrink=10,stretch=10]{microtype}
\usepackage{amsmath,amssymb,amsthm}
\usepackage[british]{babel}
\usepackage{graphicx}
\usepackage{booktabs}
\usepackage{calc}
\usepackage[hidelinks]{hyperref}

\renewcommand{\tilde}[1]{\widetilde{#1}}

\newcommand{\emphMath}[1]{\boldsymbol{\mathbf{#1}}}

\makeatletter
\def\@endtheorem{\endtrivlist}
\makeatother
\newtheorem{theorem}{Theorem}
\newtheorem{lemma}[theorem]{Lemma}
\newtheorem{proposition}[theorem]{Proposition}
\newtheorem{corollary}[theorem]{Corollary}

\newtheorem{example}{Example}

\hypersetup{
    pdftitle={On nested code pairs from the Hermitian curve},
    pdfauthor={René Bødker Christensen and Olav Geil}
}

\captionnamefont{\footnotesize\bfseries}	
\captiontitlefont{\footnotesize}					
\captiondelim{. }
\setsecheadstyle{\bfseries\large}

\counterwithout{section}{chapter}

\newcommand{\ph}{\phantom{0}}
\newcommand{\sep}{\unskip,\ }
\newcommand{\orcidID}[1]{
    \unskip\hspace{0.1em}\href{https://orcid.org/#1}{\raisebox{-0.12\height}{\includegraphics[height=0.8em]{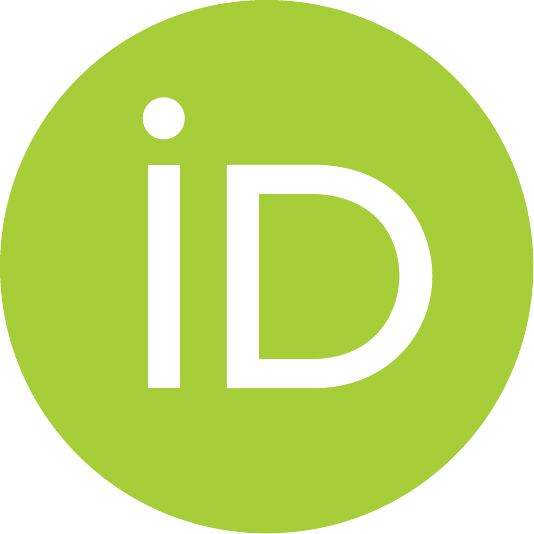}}}
}

\begin{document}
{
    \renewcommand\thefootnote{}
    \setlength{\footmarkwidth}{0em}
    \setlength{\footmarksep}{0em}
    \footnote{This manuscript version is made available under the \href{https://creativecommons.org/licenses/by-nc-nd/4.0/}{CC-BY-NC-ND 4.0 license}. The published version can be found at \url{https://doi.org/10.1016/j.ffa.2020.101742}}
    \addtocounter{footnote}{-1}
}
\begin{center}
  {\Large\fontfamily{ppl}\fontseries{b}\selectfont On nested code pairs from\\\smallskip the Hermitian curve}\\[1.5em]
  {René Bødker Christensen\orcidID{0000-0002-9209-3739}and Olav Geil\orcidID{0000-0002-9666-3399}}
  
  \bigskip{Department of Mathematical Sciences,
      Aalborg University, Denmark.}
  
  {\texttt{\{rene,olav\}@math.aau.dk}}

  \vglue 2em
  {\small\textbf{Abstract}}\\[0.5em]
  \begin{minipage}{0.9\linewidth}
    \rule{\linewidth}{0.5pt}
    \footnotesize Nested code pairs play a crucial role in the construction of ramp secret sharing
    schemes~\cite{kurihara} and in the CSS construction of quantum codes~\cite{kkk}.
    The important parameters are (1) the codimension, (2) the
    relative minimum distance of the codes, and (3) the relative minimum distance of the
    dual set of codes. Given values for two of them, one aims at finding a
    set of nested codes having parameters with these values and with the
    remaining parameter being as large as possible. In this work we study nested codes from the Hermitian
    curve. For not too small codimension, we present
    improved constructions and provide closed formula
    estimates on their performance. For small codimension we show how to choose pairs of
    one-point algebraic geometric codes in such a way that one of the relative minimum distances is
    larger than the corresponding non-relative minimum distance.


    \rule{\linewidth}{0.5pt}
  \end{minipage}
\end{center}

\section{Introduction}
In this paper we study improved constructions of nested code pairs
from the Hermitian curve. Here the phrase `improved
    construction' refers to optimizing those parameters important for the corresponding linear ramp secret sharing schemes
as well as stabilizer asymmetric quantum codes. Our work is a natural
continuation of~\cite{8048519}, where improved constructions of nested
code pairs were defined from Cartesian product point sets. The
analysis in the present paper includes a closed formula estimate on the
dimension of order bound improved Hermitian codes, which is of interest
in its own right, i.e.\ also without the above mentioned applications.

A linear ramp secret sharing scheme is a cryptographic method to
encode a secret message in ${\mathbb{F}}_q^\ell$ 
into $n$ shares from ${\mathbb{F}}_q$. These shares are then
distributed
among a group of $n$ parties and only specified subgroups are able to
reconstruct the secret. A secret sharing scheme is characterized by
its privacy number $t$ and its reconstruction number $r$. The first is
defined as the largest number such that no subgroup of this size can
obtain any information on the secret. The second is defined to be the
smallest number such that any subgroup of this size can reconstruct
the entire secret. A linear ramp secret sharing scheme can be
understood as the following coset construction. Consider linear codes $C_2\subset C_1
\subseteq {\mathbb{F}}_q^n$. Let $\{\vec{b}_1, \ldots ,
\vec{b}_{k_2}\}$ be a basis for $C_2$ and extend it to a basis
$\{\vec{b}_1, \ldots , \vec{b}_{k_2}, \vec{b}_{k_2+1}, \ldots ,
\vec{b}_{k_2+\ell}\}$ for $C_1$. Here, of course, $\ell$ is the
codimension of $C_1$ and $C_2$. Choose elements $a_1,\ldots , a_{k_2}$
uniformly and independent at random and encode the secret $\vec{s}=(s_1, \ldots ,
s_\ell)$ as the codeword $\vec{c}=(c_1, \ldots , c_n)=a_1\vec{b}_1 + \cdots
+a_{k_2}\vec{b}_{k_2}+s_1\vec{b}_{k_2+1}+\cdots +s_\ell
\vec{b}_{k_2+\ell}$. Then use $c_1, \ldots , c_n$ as the
shares.
The crucial parameters for the
construction are the codimension of the pair of nested codes and
their relative minimum distances $d(C_1,C_2)$ and $d(C_2^\perp,
C_1^\perp)$. Recall that these are defined as
\begin{align*}
d(C_1,C_2)=\min \{ w_H(\vec{c}) \mid \vec{c} \in C_1 \backslash C_2\}
\end{align*}
and similar for the dual codes.
The following well-known theorem (see for instance~\cite{kurihara}) shows how to determine the privacy
number and the reconstruction number.
\begin{theorem}\label{thesec}  
Given ${\mathbb{F}}_q$-linear codes $C_2 \subset C_1$ of length $n$
and codimension
$\ell$, the corresponding ramp secret sharing scheme encodes secrets
$\vec{s} \in {\mathbb{F}}_q^\ell$ into a set of $n$ shares from
${\mathbb{F}}_q$. The privacy number equals $t=d(C_2^\perp,
C_1^\perp)-1$, and the reconstruction number is $r=n-d(C_1,C_2)+1$.
\end{theorem}

A linear $q$-ary asymmetric quantum error-correcting code is a $q^k$
dimensional subspace of the Hilbert space ${\mathbb{C}}^{q^n}$ where
error bases are defined by unitary operators $Z$ and $X$, the first
representing phase-shift errors, and the second representing bit-flip
errors~\cite{q7,steane1996simple,kkk}. In~\cite{aq7} it was identified that in some realistic models
phase-shift errors occur  more frequently than bit-flip errors,
and the asymmetric codes were therefore introduced~\cite{aq7,aq6,ezerman2013css,aq2,MR2683447} to balance
the error correcting ability accordingly. For such codes we write the
set of parameters as $[[n,k,d_z/d_x]]_q$ where $d_z$ is the minimum
distance related to phase-shift errors and $d_x$ is the minimum
distance related to bit-flip errors.
The CSS construction transforms a
pair of nested classical linear codes $C_2 \subset C_1 \subseteq {\mathbb{F}}_q^n$
into an asymmetric quantum code. From~\cite{aq6} we have
\begin{theorem}\label{thm:css}
Consider linear codes $C_2 \subset C_1 \subseteq
{\mathbb{F}}_q^n$. Then the corresponding asymmetric quantum code defined using the
CSS construction has parameters
\begin{align*}
[[n,\ell=\dim C_1-\dim C_2,d_z/d_x]]_q
\end{align*}
where $d_z=d(C_1,C_2)$ and $d_x=d(C_2^\perp, C_1^\perp)$.
\end{theorem}
Quantum codes with $d(C_1,C_2) > d(C_1)$ or $d(C_2^\perp, C_1^\perp) >
d(C_2^\perp)$ are called impure, and they are desirable due to the fact
that one can take advantage of this property in connection with the
error-correction. More precisely, one can tolerate 
$\lfloor (d(C_1,C_2)-1)/2\rfloor$ phase-shift errors and $\lfloor
(d(C_2^\perp , C_1^\perp)-1)/2\rfloor$ bit-flip errors, respectively,
but in the decoding algorithms it is only necessary to correct up to $\lfloor (d(C_1)-1)/2\rfloor$ and
$\lfloor (d(C_2^\perp)-1)/2\rfloor$ errors, respectively. Despite this observation, only few impure codes have been presented in the literature.

With the above two applications in mind, the challenge is to find nested codes $C_2\subset C_1$ such that two of the parameters $\ell$, $d(C_1,C_2)$, $d(C_2^\perp, C_1^\perp)$ attain given prescribed values, and the remaining parameter is as large as possible.
In this paper we analyse two good
constructions from the Hermitian function field. In the first we consider code pairs such that $C_1$ is an
order bound improved primary code~\cite{AG,MR2831617} and such that $C_2$ is the dual of an order
bound improved dual code~\cite{handbook}. Considering in this case the minimum
distances rather than the relative distances is no restriction due to
the optimized choice of codes -- the minimum distances $d(C_1)$ and
$d(C_2^\perp)$ being so good that there is essentially no room for
$d(C_1,C_2) > d(C_1)$ or $d(C_2^\perp , C_1^\perp) > d(C_2^\perp)$ to hold. For this construction to work, the codimension
cannot be very small. For small codimension when $d(C_1)$ and $d(C_2^\perp)$
are far from each other we then show how to choose
ordinary one-point algebraic geometric codes such that one of the relative
distances becomes larger than the corresponding ordinary minimum
distance. In particular, this construction leads to impure
asymmetric quantum codes. 

The paper is organized as follows. In Section~\ref{sec2} we collect
material from the literature on how to determine parameters of primary
and dual codes
coming from the Hermitian curve, and we introduce the order bound improved
codes\footnote{This section also contains a collective treatment of
  the order bounds 
for general primary as well as dual (improved) one-point algebraic geometric
codes which may not be easy to find in the literature.}. In Section~\ref{sec:dimImprovedCodes} we establish closed formula lower bounds
on the dimension of order bound improved Hermitian codes of any
designed minimum distance. We then continue in Section~\ref{sec:inclusion} by
determining the pairs $(\delta_1$, $\delta_2) \in \{1, \ldots , n\}
\times \{1, \ldots , n\}$ for which the order
bound improved primary code $C_1$ of designed distance $\delta_1$ contains $C_2$, the dual of an order bound
improved dual code of designed distance $\delta_2$. This and the information from
Section~\ref{sec:dimImprovedCodes} is then translated into information on improved
nested code pairs of not too small codimension in Section~\ref{sec4.5}. Next, in
Section~\ref{sec:constr-with-relat} we determine parameters of nested one-point
algebraic geometric code pairs of
small codimension for which one of the relative distances is larger than
the non-relative.
Finally, in Section~\ref{sec:comp-with-exist} samples of the given
constructions are compared with known asymmetric quantum codes,
with existence bounds on asymmetric quantum codes, and with
non-existence bounds on linear ramp secret sharing
schemes. Section~\ref{sec:concl} is the conclusion.
\section{Hermitian codes and their
  parameters}\label{sec2}
Given an algebraic function field over a finite field, let $=P_1,
\ldots , P_n,Q$ be rational places. By $H(Q)$ we denote the
Weierstrass semigroup of $Q$, and we write
\begin{equation*}
H^\ast(Q)=\{\lambda \in H(Q) \mid C_{\mathcal{L}}(D,\lambda Q) \neq
C_{\mathcal{L}}(D,(\lambda -1)Q)\}
\end{equation*}
where $D=P_1+ \cdots +P_n$. Recall that the dual code of
$C_{\mathcal{L}}(D,\lambda Q)$ is written $C_\Omega(D,\lambda Q)$. The
order bound~\cite{handbook,MR2588125} then tells us that if
\begin{equation*}
  \vec{c} \in C_\Omega (D,(\lambda-1) Q)\, \backslash \,
  C_\Omega(D, \lambda Q)
\end{equation*}
(which can only happen if $\lambda \in H^\ast(Q)$),
then the Hamming weight of $\vec{c}$ satisfies 
\begin{equation}
w_H(\vec{c}) \geq \mu
(\lambda) \label{eqduord}
\end{equation}
where
\begin{equation*}
  \mu(\lambda) = \# \{ \eta \in H(Q) \mid \lambda-  \eta \in H(Q)\}.
\end{equation*}
The similar bound for the primary case~\cite{normtrace,AG,MR2831617} tells us that if
\begin{equation*}
  \vec{c} \in C_{\mathcal{L}}(D,\lambda Q) \, \backslash \,
  C_{\mathcal{L}}(D, (\lambda-1)Q),
\end{equation*}
then 
\begin{equation}
  w_H(\vec{c}) \geq \sigma (\lambda)\label{eqprord}
\end{equation}
where
\begin{equation*}
  \sigma (\lambda)=\# \{ \eta \in H^\ast (Q) \mid \eta - \lambda \in
H(Q)\}.
\end{equation*}
Besides implying that 
\begin{align}
  d(C_{\mathcal{L}}(D,\lambda Q)) &\geq \min \{ \sigma (\gamma) \mid 0
\leq \gamma \leq \lambda, \gamma \in H^\ast (Q)\}\nonumber\\
  d(C_{\Omega}(D,\lambda Q)) &\geq \min \{ \mu (\gamma) \mid \lambda <
\gamma , \gamma \in H^\ast (Q)\},\label{eq:boundHermitianDual}
\end{align}
which are both as strong as the Goppa bound, it tells us that for
$\epsilon , \lambda  \in H^\ast (Q)$ with $\epsilon < \lambda$ it
holds that 
\begin{equation}
d(C_{\mathcal{L}}(D,\lambda Q), C_{\mathcal{L}}(D, \epsilon
Q))\geq \min \{\sigma(\gamma) \mid \epsilon  < \gamma \leq \lambda, \gamma
\in H^\ast (Q)\},\label{eqrel1}
\end{equation}
and similarly 
\begin{equation}
d(C_\Omega (D,\epsilon Q),C_\Omega (D, \lambda Q)) \geq \min \{\mu
(\gamma) \mid \epsilon < \gamma \leq \lambda, \gamma \in H^\ast
(Q)\}.\label{eqrel2}
\end{equation}
Furthermore, for $i\in H^\ast(Q)$ let $f_i \in {\mathcal{L}}(iQ)
\backslash {\mathcal{L}}((i-1)Q)$. Then we obtain the improved primary
code
\begin{equation*}
  \tilde{E}(\delta)={\mbox{Span}}\{ (f_i(P_1), \ldots , f_i(P_n)) \mid
  \sigma(i) \geq \delta\},
\end{equation*}
which clearly has minimum distance at least $\delta$ and highest
possible dimension for a primary code with that designed distance. Similarly, the
improved dual code
\begin{equation*}
\tilde{C}(\delta)=\big( {\mbox{Span}}\{ (f_i(P_1), \ldots ,
f_i(P_n))\mid \mu(i) < \delta\} \big)^\perp
\end{equation*}
has minimum distance at least $\delta$ and again the highest possible
dimension for a dual code with that designed distance.

Turning to the Hermitian curve $x^{q+1}-y^q-y$ over
${\mathbb{F}}_{q^2}$ where $q$ is a prime power, it is well-known that
the corresponding function field has exactly $q^3+1$ rational places $P_1, \ldots ,
P_{q^3},Q$. Choosing $n=q^3$ one obtains $H(Q)=\langle q,q+1\rangle$ and 
\begin{align}
  H^\ast(Q)=\{ i q+j(q+1) \mid 0 \leq i \leq q^2-1, 0 \leq j \leq
q-1\}.\label{eq:hAstStructure}
\end{align}
In~\cite{stichtenothhermitian} it was shown that 
\begin{equation}
C_{\mathcal{L}}(D,\lambda
Q) = C_\Omega(D,(q^3+q^2-q-2-\lambda) Q) \label{eqdupr}
\end{equation} 
for any $\lambda \in H^\ast (Q)$, and the minimum distance was established for
dimensions up to a certain value. The minimum distance for the remaining dimensions was then settled
in~\cite{yanghermitian}. In the present paper we shall need improved code
constructions, and we will in some cases also
be occupied with the relative distances rather than minimum
distances. To this end we recall material
from~\cite{normtrace} on the functions $\mu$ and $\sigma$ -- stated
there in the more general case of norm-trace curves, but adapted here
to the Hermitian case. 
\begin{proposition}\label{prober}
Consider the Hermitian curve. For $iq+j(q+1) \in H^\ast(Q)$ we have
\begin{equation}
\sigma(iq+j(q+1))=\left\{ \begin{array}{ll}
q^3-iq-j(q+1)& {\mbox{ if }} 0 \leq i < q^2-q\\
(q^2-i)(q-j)& {\mbox{ if }} q^2-q \leq i \leq q^2-1,\label{eqxpression}
\end{array}
\right.
\end{equation}
and $\mu\big( (q^2-1-i)q+(q-1-j)(q+1) \big)=\sigma
\big(iq+j(q+1)\big)$. For each $\lambda \in H^\ast(Q)$ there exists a
word $\vec{c} \in \big( C_{\mathcal{L}}(D, \lambda Q)\backslash
C_{\mathcal{L}}(D,(\lambda-1)Q) \big) \cap \tilde{E}(\sigma(\lambda))$ having
Hamming weight equal to $\sigma(\lambda)$.
\end{proposition}
\begin{proof}
Given a numerical semigroup $\Lambda$ with finitely many gaps and
an element $\lambda \in \Lambda$,  we know
from~\cite[Lem.\ 5.15]{handbook} that $\# \big( \Lambda
\backslash (\lambda + \Lambda) \big)=\lambda$. As $\# H^\ast (Q)=q^3$
we therefore obtain $\sigma(iq+j(q+1)) \geq q^3-(iq+j(q+1))$. On the other hand, it is
clear that $\sigma(iq+j(q+1)) \geq (q^2-i)(q-j)$ by the definition of $\sigma$. Taking the
maximum between these two expressions, we obtain the right hand side
of~(\ref{eqxpression}). That
these estimates on $\sigma$ are the true values and that the last part
of the proposition holds true both follow as a consequence of~\cite[Lem.\
4]{normtrace}. The details of applying ~\cite[Lem.\ 4]{normtrace} are
left for the reader. 
Finally, the relation between $\mu$ and
$\sigma$ is a consequence of $H^\ast (Q)$ being a box in the
parameters $i$ and $j$, see \eqref{eq:hAstStructure}. 
\end{proof}
In \ref{app:addit-results-sigma} we list a series of lemmas which
all follow as corollaries to Proposition~\ref{prober} and which will
be needed in Sections~\ref{sec:dimImprovedCodes} and
\ref{sec:inclusion}.

Throughout the rest of the paper we restrict to considering codes
derived from the Hermitian curve, and we always assume the length to be
$n=q^3$. From Proposition~\ref{prober} we see that the bound~(\ref{eqrel1})
 on the relative distance of  $C_{\mathcal{L}}(D,\epsilon Q)
\subset C_{\mathcal{L}}(D, \lambda Q)$ is sharp. A similar remark then
holds for the bound~(\ref{eqrel2})
on the dual codes due to~(\ref{eqdupr}). Finally, we observe from~\cite[Sec.\ 4]{normtrace}
that 
\begin{equation}
\tilde{E}(\delta)=\tilde{C}(\delta)\label{eqimpens}
\end{equation}
holds. Proposition~\ref{prober} therefore
not only gives us the true value of the minimum distance of the
improved primary
codes (without loss of generality we may assume $\delta = \sigma(\lambda)$ for some $\lambda \in H^\ast
(Q)$), but also does it for the improved dual codes.

We conclude the section with some information on the cases where the improved
primary codes coincide with one-point algebraic geometric codes.
\begin{corollary}\label{cor:higherDim}
  For $\delta > q^2-q$ we have $\tilde{E}(\delta)=C_{\mathcal{L}}\big(D,(q^3-\delta)Q\big)$, but $C_{\mathcal{L}}\big(D,(q^3-(q^2-q))Q\big)$ is strictly contained in $\tilde{E}(q^2-q)$.
\end{corollary}
This corollary implies that the dimension of $\tilde{E}(\delta)$ can be determined from the usual one-point Hermitian codes whenever $\delta>q^2-q$. For later reference we state these dimensions in terms of $\delta$.
\begin{proposition}\label{prop:dimUsualHermitian}
  Denote by $g=q(q-1)/2$ the genus of Hermitian function field. If $q^2-q<\delta<q^3-2g+2$, then the dimension of $\tilde{E}(\delta)$ is given by $q^3-g+1-\delta$.
  If $q^3-2g+2\leq\delta\leq q^3$, we have
  \begin{equation*}
    \dim\tilde{E}(\delta)=\sum_{s=0}^{a+b}(s+1)-\max\{a,0\}
  \end{equation*}
  where $q^3-\delta=aq+b(q+1)$ for $-q<a<q$ and $0\leq b<q$.
\end{proposition}
\begin{proof}
  First, note that in both cases Corollary~\ref{cor:higherDim} implies the equality $\tilde{E}(\delta)=C_\mathcal{L}(D,(q^3-\delta)Q)$.
  For the first case recall from \cite[Cor.\ II.2.3]{stichtenoth} that the code $C_\mathcal{L}(D,\lambda G)$ has dimension $\lambda+1-g$ whenever $2g-2<\lambda<n$.
  By the assumptions on $\delta$, we have $q^3-\delta>q^3-(q^3-2g+2)=2g-2$, meaning that $C_\mathcal{L}(D,(q^3-\delta)Q)$ has dimension $q^3-\delta+1-g$. By the observation in the beginning of the proof, the same holds true for $\tilde{E}(\delta)$.

  To prove the second case, observe that the dimension of $C_\mathcal{L}(d,(q^3-\delta)Q)$ is exactly the number of elements $\lambda$ in $H^\ast(Q)$ with $\lambda\leq q^3-\delta$. By \eqref{eq:hAstStructure} such elements have the form $\lambda=iq+j(q+1)$, but equivalently we can write $\lambda=i'q+j$ where $i'=i+j$. By the division algorithm this representation is unique for $0\leq j<q$. For $i'q+j$ to satisfy the requirements of \eqref{eq:hAstStructure}, we also require $i'-j=i\geq 0$. That is, $H^\ast(Q)$ contains the integers whose quotients modulo $q$ are at least their remainders modulo $q$.

  Writing $q^3-\delta=(a+b)q+b$ with $0\leq b<q$ using the division algorithm, the number of elements in $H^\ast(Q)$ less than $(a+b+1)q$ is given by $\sum_{s=0}^{a+b}(s+1)$.
  If $b\geq a+b$, which happens only if $a\leq 0$, this number is also
  the number of elements with value at most $q^3-\delta$. Otherwise,
  the count includes $b-(a+b)=a$ elements of $H^\ast(Q)$ that are
  greater than $q^3-\delta$. Hence, subtracting $\max\{a,0\}$ gives
  the desired count in both cases. 
\end{proof}
It remains to establish information on the dimension of $\tilde{E}(\delta)$ for
$\delta \leq q^2-q$ since the improved codes differ from the usual Hermitian codes in this case. This subject is treated in the next section.
\section{The dimension of improved codes}\label{sec:dimImprovedCodes}
As explained in the previous section, the dimension of
$\tilde{E}(\delta)$ can be determined from well-known methods as
long as $\delta>q^2-q$. In this section we present closed formula
lower bounds on the dimension in the remaining cases. We start with an
important lemma.
\begin{lemma}\label{lem:integerPointsIntegral}
  Let $\delta \leq q^2$. The number of integer points 
  \begin{align*}
    (x,y)\in \{q^2-q, \ldots ,q^2-1\} \times \{0, \ldots , q-1\}
  \end{align*}
  with $(q^2-x)(q-y) \geq \delta$ is at least
  \begin{equation}
    q^2-\left\lfloor\delta+\delta\ln ( q^2/\delta)\right\rfloor. \label{eqsnabel1}
  \end{equation}
  If $\delta < q$, then the number of integer points is at least
  \begin{equation}
    q^2-\left\lfloor\delta+\delta \ln (\delta)\right\rfloor, \label{eqsnabel2}
  \end{equation}
  which is stronger than (\ref{eqsnabel1}).
\end{lemma}
\begin{proof}
  The number of integer points in the given Cartesian product is at
  least that of the volume of
  \begin{align*}
    \{(x,y) \in [q^2-q,q^2] \times [0,q] \mid (q^2-x)(q-y) \geq\delta\},
  \end{align*}
  which equals
  \begin{align*}
    &\int_{q^2-q}^{q^2-\frac{\delta}{q}} \int_{0}^{q-\frac{\delta}{q^2-x}}
       \, dy\, dx \\
    =&q(q^2-\frac{\delta}{q}-q^2+q)+\int_{q^2-\frac{\delta}{q}}^{q^2-q}\frac{\delta}{q^2-x}
        \, dx \\
    =&q^2-\delta -\delta [ \ln
        (z)]^q_{\frac{\delta}{q}}=q^2-\delta-\delta \ln (q^2/\delta),
  \end{align*}
  where we used the substitution $z=q^2-x$. Since the number of integer points is integral, we obtain the bound $\lceil q^2-\delta-\delta\ln(q^2/\delta)\rceil$, which is the same as \eqref{eqsnabel1}.
  
  If $\delta < q$, then the number of integer points is at least the combined volumes of
  \begin{align*}
    \{(x,y) \in [q^2-q,q^2]\times [0,q] \mid x \leq q^2-\delta {\mbox{
    or }} y \leq q-\delta \}
  \end{align*}
  and
  \begin{align*}
    \{(x,y) \in [q^2-\delta, q^2] \times [q-\delta,q] \mid (q^2-x)(q-y)
    \geq \delta\}.
  \end{align*}
The first mentioned volume equals $q^2-\delta^2$. The latter volume is
  \begin{align*}
    &\int_{q^2-\delta}^{q^2-1}\int_{q-\delta}^{q-\frac{\delta}{q^2-x}} \ dy
       \, dx \\
    =&\int_{q^2-\delta}^{q^2-1}\bigg(\delta -\frac{\delta}{q^2-x}\bigg)\,
        dx\\
    =&\delta (\delta-1)-\delta[\ln (z)]_1^\delta=\delta(\delta-1)-\delta
        \ln (\delta).
  \end{align*}
  Adding up the two volumes, we obtain \eqref{eqsnabel2} by applying the ceiling function as above.  
\end{proof}
The dimension of the improved codes of designed distance at most $q^2-q$ is covered by the following two
propositions. Recall from~(\ref{eqimpens}) that the equality $\tilde{C}(\delta)=\tilde{E}(\delta)$ holds for
codes defined from the Hermitian function field. Hence, the stated
formulas for primary codes also hold for the dual codes.
\begin{proposition}\label{prop:dimTricky}
  Given $q < \delta \leq q^2-q$ write 
  \begin{align*}
    q^3-\delta = q^3-q^2+aq+b(q+1)
  \end{align*}
  where $-q < a < q$ and $0 \leq b < q$. \\
  If $0 < a$, then 
  \begin{align*}
    \dim (\tilde{E}(\delta)) \geq q^3-\delta - g+1 -
    \sum_{s=0}^{a+b}(s+1) + a + q^2-\left\lfloor\delta+\delta\ln ( q^2/\delta)\right\rfloor.
  \end{align*}
  If $a\leq 0$, then 
  \begin{align*}
    \dim (\tilde{E}(\delta)) \geq q^3-\delta - g+1 -
    \sum_{s=0}^{a+b}(s+1) + q^2-\left\lfloor\delta+\delta\ln ( q^2/\delta)\right\rfloor.
  \end{align*}
\end{proposition}
\begin{proof}
  Let $g=q(q-1)/2$ be the number of gaps in
  $H(Q)$, i.e.\ the genus of the function field. As is well-known, for $2g \leq \lambda < q^3-1$ the number of $\epsilon \in H^\ast(Q)$ with $\epsilon \leq \lambda$ equals $\lambda-g+1$. Therefore, by choosing
  $\lambda=q^3-\delta$ the restriction on $\delta$ as given in the
  proposition implies that there are exactly $q^3-\delta-g+1$ elements
  $\epsilon \in H^\ast(Q)$ with $\epsilon \leq q^3-\delta$. These elements then
  satisfy $\sigma(\epsilon) \geq \delta$ by Lemma~\ref{corstronger}.
  From~(\ref{eqxpression}) we see that the additional elements in $H^\ast(Q)$ with
  $\sigma(\epsilon) \geq \delta$ must belong to \begin{equation}
  \{iq+j(q+1) \mid q^2-q \leq i \leq q^2-1, \,  0 \leq j\leq q-1
    \}. \label{eqdimondius}
 \end{equation}
  Lemma~\ref{lem:integerPointsIntegral} gives an estimate on the total number of elements $\epsilon$ in~(\ref{eqdimondius}) with $\sigma(\epsilon) \geq \delta$. Adding this number to
  $q^3-\delta-g+1$, we have counted the elements $\epsilon$
  in~(\ref{eqdimondius}) with $\epsilon \leq q^3-\delta$ twice. By using similar arguments as in the proof of Proposition \ref{prop:dimUsualHermitian}, the number of such elements equals
  $\sum_{s=0}^{a+b}(s+1)-a$ when $0 \leq a < q$, and it equals
  $\sum_{s=0}^{a+b}(s+1)$ when $-q < a < 0$. This proves the proposition.
\end{proof}
\begin{proposition}\label{prop:dimNotTricky}
  Given $1\leq\delta\leq q$ the dimension of the code $\tilde{E}(\delta)$ satisfies
  \begin{align*}
    \dim(\tilde{E}(\delta))\geq q^3-\left\lfloor\delta+\delta\ln(\delta)\right\rfloor.
  \end{align*}
\end{proposition}
\begin{proof}
  By Lemma~\ref{lem:smallSigma} the elements $\lambda \in H^\ast (Q)$ which do not satisfy
  $\sigma(\lambda)\geq\delta$ must belong to $\{iq+j(q+1) \mid
  q^2-\delta\leq i\leq q^2, q-\delta\leq j < q\}$. 
  The number of elements in this set having $\sigma(\lambda)<\delta$ is bounded above by $\lfloor\delta+\delta\ln(\delta)\rfloor$ by Lemma \ref{lem:integerPointsIntegral}. Since the total number of monomials in $H^\ast(Q)$ is $q^3$, the result follows.
\end{proof}

\section{Inclusion of improved codes}\label{sec:inclusion}
As already mentioned our first construction of improved nested code pairs consists of choosing $\tilde{C}(\delta_2)$ and $\tilde{E}(\delta_1)$ such that $\tilde{C}(\delta_2)^\perp\subset\tilde{E}(\delta_1)$. To treat this construction we therefore need a clear picture of the pairs $(\delta_1,\delta_2)$ of minimum distances that imply this inclusion. We establish this in the present section.
As it turns out, the formulas for $\sigma$ and $\mu$ given in Proposition \ref{prober} mean that several cases must be considered, and each case is presented as a separate proposition.

In the following, quantifiers on $\lambda,\varepsilon$ are considered on the domain $H^\ast(Q)$.
Given $\delta_1\in\sigma(H^\ast(Q))$, define $\delta_2^{\max}$ to be the maximal value of $\delta_2$ such that $\tilde{C}(\delta_2)^\perp\subseteq\tilde{E}(\delta_1)$ holds. This inclusion is equivalent to
\begin{align}\label{eq:logicInclusion}
  \forall\lambda\colon \big[(\sigma(\lambda)<\delta_1)\rightarrow(\mu(\lambda)\geq\delta_2)\big].
\end{align}
We first observe that if we can find a $\lambda_1\in H^\ast(Q)$ such that
\begin{align}\label{eq:inclusionLowCondition}
  \big[\forall\varepsilon>\lambda_1\colon\:\mu(\varepsilon)\geq\mu(\lambda_1)\big] \wedge
  \big[\forall\varepsilon<\lambda_1\colon\: \sigma(\varepsilon)\geq\delta_1\big]
\end{align}
is true, then \eqref{eq:logicInclusion} is also true whenever $\delta_2\leq\mu(\lambda_1)$. In particular, we therefore have
\begin{align}\label{eq:deltaMaxLower}
  \delta_2^{\max}\geq\mu(\lambda_1).
\end{align}
On the other hand, we immediately see from \eqref{eq:logicInclusion} that a $\lambda_2\in H^\ast(Q)$ with
\begin{align}\label{eq:inclusionUppCondition}
  \sigma(\lambda_2)<\delta_1
\end{align}
implies the bound
\begin{align}\label{eq:deltaMaxUpper}
  \delta_2^{\max}\leq\mu(\lambda_2).
\end{align}
In the proofs of each of the following propositions, our strategy 
therefore is to determine $\lambda_1$ and $\lambda_2$ satisfying \eqref{eq:inclusionLowCondition} and \eqref{eq:inclusionUppCondition}, respectively, while also satisfying $\mu(\lambda_1)=\mu(\lambda_2)$. From \eqref{eq:deltaMaxLower} and \eqref{eq:deltaMaxUpper} it then follows that $\delta_2^{\max}=\mu(\lambda_1)$. Note, however, that $\lambda_1$ and $\lambda_2$ need not be distinct. If $\lambda_1=\lambda_2$, we shall simply use $\lambda$.

With this strategy in mind, the following lemmas will prove very useful.
\begin{lemma}\label{lem:sigmaOrderedPartial}
  Let $\lambda=iq+j(q+1)\in H^\ast(Q)$, meaning that $0\leq i<q^2$ and $0\leq j<q$. In addition, assume that $i\leq q^2-q$, $i=q^2-1$, or $j=0$. If $\varepsilon\in H^\ast(Q)$ satisfies $\varepsilon<\lambda$, then $\sigma(\varepsilon)\geq\sigma(\lambda)$.
\end{lemma}
\begin{proof}
  For both the cases $i< q^2-q$ and $j=0$, the claim follows by Lemma~\ref{lem1}. If $i=q^2-q$, Lemma~\ref{lem:symmetricCorner} implies that $\sigma(\lambda)=\sigma((i+j)q)$, and any $\varepsilon\in H^\ast(Q)$ satisfying $\lambda>\varepsilon>(i+j)q$ has $\sigma(\varepsilon)\geq\sigma((i+j)q)$ by Lemma~\ref{lem:improvedCorner}. This also holds true for $\varepsilon<(i+j)q$ by the first part of the proof.
  
  Finally, for $i=q^2-1$ consider $\varepsilon=i'q+j'(q+1)\in H^\ast(Q)$ with $\varepsilon<\lambda$. If $q^2-q\leq i' \leq q^2-1$, the claim follows by Lemmas~\ref{lem:sequence} and \ref{lem:improvedCorner}. Otherwise, $\varepsilon$ is at most $(q^2-q-1)q+(q-1)(q+1)=q^3-q-1$, meaning that $\sigma(\varepsilon)\geq q+1$ by Lemma~\ref{lem1}. The proof follows by noting that $\sigma(\lambda)\leq q$.
\end{proof}
\begin{lemma}\label{lem:muOrderedPartial}
  Let $\lambda=iq+j(q+1)\in H^\ast(Q)$, meaning that $0\leq i<q^2$ and $0\leq j<q$. In addition, assume that $i\geq q-1$ or $j=0$. If $\varepsilon\in H^\ast(Q)$ satisfies $\varepsilon>\lambda$, then $\mu(\varepsilon)\geq\mu(\lambda)$.
\end{lemma}
\begin{proof}
  This proof is similar to the one of Lemma~\ref{lem:sigmaOrderedPartial}. Defining the notation $\lambda'=(q^2-1-i)q+(q-1-j)(q+1)$, the translation of Lemma~\ref{corstronger} into information on $\mu$ gives $\mu(\lambda)=n-\lambda'$ whenever
  \begin{equation}\label{eq:muOrderCondition}
    i>q-1\text{ or }j=q-1.
  \end{equation}
  Additionally, if $\varepsilon\in H^\ast(Q)$ with $\varepsilon>\lambda$, then $\varepsilon'<\lambda'$ where the $\varepsilon'$ is defined in the same way as $\lambda'$. This implies that $\mu(\lambda)<\mu(\varepsilon)$ when $\lambda=iq+j(q+1)$ satisfies \eqref{eq:muOrderCondition}. This immediately proves the claim for $i>q-1$.

  For $i=q-1$ the $\mu$-equivalent of Lemma~\ref{lem:symmetricCorner} gives $\mu(\lambda)=\mu((i-(q-1-j))+(q-1)(q+1))$, and any elements between have greater $\mu$-value by the translation of Lemma~\ref{lem:improvedCorner}. The remaining elements greater than $\lambda$ are covered by the first part of the proof.

  The last part of the proof is $j=0$ and $i<q-1$, which follows the same procedure as the last part of the proof of Lemma \ref{lem:sigmaOrderedPartial}.
\end{proof}

\begin{proposition}\label{prop:inclusionRightCorner}
  Let $2\leq\delta_1\leq q$. Then $\tilde{C}(\delta_2)^\perp\subseteq\tilde{E}(\delta_1)$ if and only if $\delta_2\leq q^3-(\delta_1-2)(q+1)$.
\end{proposition}
\begin{proof}
  Let $\lambda'=(q^2-1)q+(q-\delta_1)(q+1)$. We have $\sigma(\lambda')=\delta_1$ by \eqref{eqxpression}, and Lemma \ref{lem:sigmaOrderedPartial} implies that $\sigma(\varepsilon)\geq\delta_1$ for all $\varepsilon<\lambda'$. Additionally, Lemmas \ref{lem:sequence} and \ref{lem:improvedCorner} yield that $\lambda=\lambda'+q+1$ is the smallest element in $H^\ast(Q)$ with a strictly smaller $\sigma$-value.
  Combining this with Lemma \ref{lem:muOrderedPartial} applied to $\lambda$ reveals that $\lambda$ satisfies \eqref{eq:inclusionLowCondition}. However, \eqref{eq:inclusionUppCondition} is satisfied as well since $\sigma(\lambda)<\delta_1$. Thus, $\delta_2^{\max}=\mu(\lambda)=q^3-(\delta_1-2)(q+1)$.
\end{proof}

\begin{proposition}\label{prop:inclusionRightMixed}
  Let $q<\delta_1\leq q^2-q$. Then $\tilde{C}(\delta_2)^\perp\subseteq\tilde{E}(\delta_1)$ if and only if
  \begin{align*}
      \delta_2\leq\begin{cases}
        q^3-q^2+q-\delta_1+2 &\text{if }0\leq b\leq a\\
        q^3-q^2-a(q+1) &\text{if }b>a
      \end{cases}
  \end{align*}
  where $\delta_1-(q+1)=aq+b$ for $a\geq 0$ and $0\leq b<q$.
\end{proposition}
\begin{proof}
  First, observe that $aq+b\leq (q-3)q+(q-1)$, meaning that $a$ is at most $q-3$.
  
  Assume that $b=0$. Then $\lambda=(q^2-1-a)q$ has $\sigma(\lambda)=\delta_1-1$, meaning that it satisfies \eqref{eq:inclusionUppCondition}. Note that $\mu(\lambda)=q^3-q^2-aq+1$, which can be rewritten to obtain the claimed expression.
  
  If $0<b\leq a$, we can use $\lambda=(q^2-q-1-a+(b-1))q+(q-1-(b-1))(q+1)$, which satisfies \eqref{eq:inclusionUppCondition} since $\sigma(\lambda)=\delta_1-1$. Here, we see that
  \begin{align*}
    \mu(\lambda)=q^3-q^2-aq-b+1=q^3-q^2+q-\delta_1+2.
  \end{align*}
  Finally, for $b>a$ we can let $\lambda=(q^2-q-1)q+(q-1-a)(q+1)$ with $\sigma(\lambda)=(a+1)q+a+1<\delta_1$. Again, $\lambda$ satisfies \eqref{eq:inclusionUppCondition}. Calculating the value of $\mu$ gives $\mu(\lambda)=q^3-q^2-a(q+1)$.
  
  In all three of the above situations, the element immediately preceding $\lambda$ in $H^\ast(Q)$ is given by $\lambda'=\lambda-1$, and the reader  may verify that $\sigma(\lambda')\geq\delta_1$. In each case applying Lemma \ref{lem:sigmaOrderedPartial} to $\lambda'$ implies that $\sigma(\varepsilon)\geq\delta_1$ for all $\varepsilon<\lambda$. Lemma \ref{lem:muOrderedPartial} applied to $\lambda$ then shows that $\lambda$ satisfies \eqref{eq:inclusionLowCondition} as well. In conclusion, the specified values of $\lambda$ satisfy both \eqref{eq:inclusionLowCondition} and \eqref{eq:inclusionUppCondition}, and computing each value of $\mu(\lambda)$ gives the expression in the proposition.
\end{proof}

\begin{proposition}\label{prop:inclusionMiddle}
  Let $q^2-q<\delta_1\leq q^3-2q^2+2q$. Then $\tilde{C}(\delta_2)^\perp\subseteq\tilde{E}(\delta_1)$ if and only if $\delta_2\leq q^3-q^2+q+2-\delta_1$.
\end{proposition}
\begin{proof}
  Set $\lambda'=n-\delta_1$ and observe that $\lambda'\geq 2q^2-2q=4g$ where $g$ is the genus of the Hermitian function field. Thus, $\lambda'$ is a non-gap in $H^\ast(Q)$, and $\sigma(\lambda')=\delta_1$ by \eqref{eqxpression}. Lemma \ref{lem:sigmaOrderedPartial} implies that any smaller element of $H^\ast(Q)$ has $\sigma$-value at least $\delta_1$. We see, however, that $\lambda=\lambda'+1$ has $\sigma(\lambda)=\delta_1-1$, and it must be the smallest such value. At the same time it meets the requirements of Lemma \ref{lem:muOrderedPartial}, implying that \eqref{eq:inclusionLowCondition} is fulfilled. As already noted $\lambda$ satisfies \eqref{eq:inclusionUppCondition} as well, meaning that $\delta_2^{\max}=\mu(\lambda)=q^3-q^2+q+2-\delta_1$. 
\end{proof}

\begin{proposition}\label{prop:inclusionLeftMixed}
  Let $q^3-2q^2+2q<\delta_1\leq q^3-q^2$. Then $\tilde{C}(\delta_2)^\perp\subseteq\tilde{E}(\delta_1)$ if and only if
  \begin{align*}
    \delta_2\leq\begin{cases}
      (a+1)q+b+2 &\text{if }b<a\\
      (a+2)q &\text{if }a\leq b<q-1\\
      (a+2)q+1 &\text{if }b=q-1
    \end{cases}
  \end{align*}
  where $q^3-q^2-\delta_1=aq+b$ for $a\geq 0$ and $0\leq b<q$.
\end{proposition}
\begin{proof}
  First note that $aq+b\leq (q-2)q$, implying that $a$ is at most $q-2$.
  Assume that $b<a$ and let $\lambda=(q+a-(b+1))q+(b+1)(q+1)$. This element satisfies the requirements of Lemma \ref{lem:muOrderedPartial}, and $\lambda-1$ satisfies the requirements of Lemma \ref{lem:sigmaOrderedPartial}. This means that $\lambda$ fulfils \eqref{eq:inclusionLowCondition}. Simultaneously, \eqref{eq:inclusionUppCondition} is met since $\sigma(\lambda)=\delta_1-1$. Thus, $\delta_2^{\max}=\mu(\lambda)=(a+1)q+b+2$.

  Let $a=b$ and $\lambda=(q-1)q+(a+1)(q+1)$. Applying Lemmas \ref{lem:sigmaOrderedPartial} and \ref{lem:muOrderedPartial} to $\lambda-1$ and $\lambda$ as above, we see that $\lambda$ satisfies \eqref{eq:inclusionLowCondition}. It also meets \eqref{eq:inclusionUppCondition} since $\sigma(\lambda)=\delta_1-1$. Subsequently, $\delta_2^{\max}=\mu(\lambda)=(a+2)q$.
  
  Now, consider $a< b<q-1$ and let $\lambda_1=(q-1)q+(a+1)(q+1)$ and $\lambda_2=(a+1)q+(q-1)(q+1)$. We can apply both Lemmas \ref{lem:sigmaOrderedPartial} and \ref{lem:muOrderedPartial} to $\lambda_1$ to obtain that it satisfies \eqref{eq:inclusionLowCondition}. On the other hand, $\sigma(\lambda_2)<\delta_1$ implies that \eqref{eq:inclusionUppCondition} is fulfilled. In addition, $\mu(\lambda_1)=\mu(\lambda_2)$, which gives $\delta_2^{\max}=\mu(\lambda_1)=(a+2)q$.
  
  The remaining part is $b=q-1$. If this happens, note that $\lambda=(q+a+1)q$ has $\sigma(\lambda)=\delta_1-1$, whereas $\lambda-1=(a+1)q+(q-1)(q+1)$ has $\sigma(\lambda-1)=\delta_1$. By the same arguments as in the first part of the proof, we obtain that $\delta_2^{\max}=\mu(\lambda)=(a+2)q+1$.
\end{proof}

\begin{proposition}\label{prop:inclusionLeftCorner}
  Let $q^3-q^2\leq\delta_1\leq q^3$. Then $\tilde{C}(\delta_2)^\perp\subseteq\tilde{E}(\delta_1)$ if and only if
  \begin{align*}
    \delta_2\leq\begin{cases}
      a+1 &\text{if }b< a\\
      a+2 &\text{if }b\geq a
    \end{cases}
  \end{align*}
  where $q^3-\delta_1=aq+b$ for $a\geq 0$ and $0\leq b<q$.
\end{proposition}
\begin{proof}
  Assume first that $b<a$, and set $\lambda_1=aq$ and $\lambda_2=a(q+1)$. The latter meets the assumptions of Lemmas \ref{lem:sigmaOrderedPartial} and \ref{lem:muOrderedPartial}, meaning that $\lambda_2$ satisfies \eqref{eq:inclusionLowCondition}. Observe that $\sigma(\lambda_1)<\delta_1$ and $\mu(\lambda_1)=\mu(\lambda_2)$. From this we see that $\delta_2^{\max}=\mu(\lambda_1)=a+1$.
  
  Otherwise, if $b\geq a$, let $\lambda=(a+1)q$, which satisfies \eqref{eq:inclusionUppCondition} by the observation that $\sigma(\lambda)<\delta_1$. The element of $H^\ast(Q)$ immediately preceding $\lambda$ is $\lambda'=a(q+1)$, which has $\sigma(\lambda')\geq\delta_1$. As in the previous proofs, applying Lemmas \ref{lem:sigmaOrderedPartial} and \ref{lem:muOrderedPartial} to $\lambda'$ and $\lambda$, respectively, shows that $\lambda$ fulfils \eqref{eq:inclusionLowCondition}. Hence, $\delta_2^{\max}=\mu(\lambda)=a+2$.  
\end{proof}
It is worth noting that $\tilde{C}(\delta_2)^\perp\subseteq\tilde{E}(\delta_1)$ if and only if every $\lambda\in H^\ast(Q)$ with $\sigma(\lambda)<\delta_1$ also satisfies $\mu(\lambda)\geq\delta_2$. By Proposition \ref{prober} this may be rewritten as $\mu(\lambda)<\delta_1$ implying $\sigma(\lambda)\geq\delta_2$. Hence, the inclusion of codes is symmetric in the sense that $\tilde{C}(\delta_2)^\perp\subseteq\tilde{E}(\delta_1)$ if and only if $\tilde{C}(\delta_1)^\perp\subseteq\tilde{E}(\delta_2)$.

One could expect that this symmetry would show up in Propositions \ref{prop:inclusionRightCorner}--\ref{prop:inclusionLeftCorner} as well. However, this is not the case since the propositions describe the \emph{maximal} value of $\delta_2$ such that $\tilde{C}(\delta_2)^\perp\subseteq\tilde{E}(\delta_1)$ for a given value of $\delta_1$. Although this implies that $\tilde{C}(\delta_1)^\perp\subseteq\tilde{E}(\delta_2)$, there may be a $\delta'>\delta_1$ such that $\tilde{C}(\delta')^\perp\subseteq\tilde{E}(\delta_2)$ as shown in Example \ref{ex:notSymmetric} below.

\begin{example}\label{ex:notSymmetric}
  Let $q=4$ and set $\delta_1=6$. Then considering the values of
  $\sigma$ and $\mu$ as given in Table \ref{tab1} of the Appendix reveals that the greatest possible value of $\delta_2$ such that $\tilde{C}(\delta_2)^\perp\subseteq\tilde{E}(\delta_1)$ is given by $\delta_2=48$.
  By the observations above we know that this implies $\tilde{C}(6)^\perp\subseteq\tilde{E}(48)$ as well. However, inspecting the tables again will reveal that the $\tilde{C}(8)^\perp$ is also a subset of $\tilde{E}(48)$. Notice that both of these observations agree with the formulas in Propositions \ref{prop:inclusionRightMixed} and \ref{prop:inclusionLeftMixed}.
\end{example}

\section{Improved nested codes of not too small codimension}\label{sec4.5}
Based on our findings in Sections~\ref{sec:dimImprovedCodes} and \ref{sec:inclusion}, we are now able to describe the parameters of our first construction of nested code pairs, namely the one where the codimension is not too small.
If $\delta_1,\delta_2\in H^\ast(Q)$ satisfy the conditions in one of
the Propositions
\ref{prop:inclusionRightCorner}--\ref{prop:inclusionLeftCorner}, it
follows that
$\tilde{C}(\delta_2)^\perp\subseteq\tilde{E}(\delta_1)$. By the bounds
\eqref{eqrel1} and \eqref{eqrel2} and the observation following
Proposition \ref{prober}, the relative distance of this code pair is
exactly $d(\tilde{E}(\delta_1))=\delta_1$, and the relative distance of its dual is $d(\tilde{C}(\delta_2)^\perp)=\delta_2$. 

For each possible pair of designed distances described in Propositions~\ref{prop:inclusionRightCorner}--\ref{prop:inclusionLeftCorner}, we can combine the dimensions of the usual Hermitian codes with the dimension bounds of Propositions \ref{prop:dimTricky} and \ref{prop:dimNotTricky}. This gives bounds on the codimension, $\ell$, of $\tilde{C}(\delta_2)^\perp\subseteq\tilde{E}(\delta_1)$.
\begin{proposition}\label{prop:codimension1}
  Let $\delta_1,\delta_2\in H^\ast(Q)$, and $\delta_1\leq q$. Further, let $\delta_2$ satisfy the conditions of Proposition \ref{prop:inclusionRightCorner}, meaning that $\tilde{C}(\delta_2)^\perp\subseteq\tilde{E}(\delta_1)$. Denote their codimension by $\ell$.

  If $\delta_2\leq q$, then
  \begin{equation*}
    \ell\geq q^3-\lfloor\delta_1+\delta_1\ln(\delta_1)\rfloor-\lfloor\delta_2+\delta_2\ln(\delta_2)\rfloor.
  \end{equation*}
  If $q<\delta_2\leq q^2-q$, then
  \begin{align*}
    \ell\geq q^3+q^2-g+1-\lfloor\delta_1+\delta_1&\ln(\delta_1)\rfloor -\delta_2-\sum_{s=0}^{a+b}(s+1)\\
    &-\lfloor\delta_2+\delta_2\ln\left(q^2/\delta_2\right)\rfloor+\max\{a,0\}
  \end{align*}
  where $a$ and $b$ are as in Proposition \ref{prop:dimTricky} applied to $\delta_2$.

  If $q^2-q<\delta_2< q^3-2g+2$, then
  \begin{equation*}
    \ell\geq q^3-g+1-\lfloor\delta_1+\delta_1\ln(\delta_1)\rfloor-\delta_2.
  \end{equation*}
  Finally, if $q^3-2g+2\leq\delta_2$, we have
  \begin{equation*}
    \ell\geq \sum_{s=0}^{a+b}(s+1)-\lfloor\delta_1+\delta_1\ln(\delta_1)\rfloor-\max\{a,0\}
  \end{equation*}
  where $q^3-\delta_2=aq+b(q+1)$ for $-q<a<q$ and $0\leq b<q$.
\end{proposition}
\begin{proof}
  By Proposition~\ref{prop:dimNotTricky} we have $\dim\tilde{E}(\delta_1)\geq q^3-\lfloor\delta_1+\delta_1\ln(\delta_1)\rfloor$. In each case, we can obtain a bound on the codimension $\ell$ by subtracting an upper bound on $\tilde{C}(\delta_2)^\perp=q^3-\dim\tilde{E}(\delta_2)$. In turn, such a bound can be obtained via a lower bound on $\dim\tilde{E}(\delta_2)$.
  
  In the case $\delta_2\leq q$ the dimension of $\tilde{E}(\delta_2)$ can again be bounded by Proposition~\ref{prop:dimNotTricky}.
  In the second case the bound on $\dim\tilde{E}(\delta_2)$ follows by Proposition~\ref{prop:dimTricky}. Proposition~\ref{prop:dimUsualHermitian} delivers the bounds in the two final cases.
\end{proof}

\begin{proposition}\label{prop:codimension2}
  Let $\delta_1,\delta_2\in H^\ast(Q)$ and $q<\delta_1\leq q^2-q$. Further, let $\delta_2$ satisfy the conditions of Proposition~\ref{prop:inclusionRightCorner}, meaning that $\tilde{C}(\delta_2)^\perp\subseteq\tilde{E}(\delta_1)$. Denote their codimension by $\ell$ and let $a_1,b_1$ be as in Proposition~\ref{prop:dimTricky} applied to $\delta_1$.

  If $\delta_2\leq q$, then
  \begin{align*}
    \ell\geq q^3+q^2-g+1-\delta_1-\sum_{s=0}^{a_1+b_1}(&s+1)+\max\{a_1,0\}\\
    &-\lfloor\delta_1+\delta_1\ln(q^2/\delta_1)\rfloor-\lfloor\delta_2+\delta_2\ln(\delta_2)\rfloor.
  \end{align*}
  If $q<\delta_2\leq q^2-q$, then
  \begin{align*}
    \ell\geq &q^3+2q^2-2g+2-(\delta_1+\delta_2)-\sum_{s=0}^{a_1+b_1}(s+1)+\max\{a_1,0\}\\
    &-\sum_{s=0}^{a_2+b_2}(s+1)+\max\{a_2,0\}-\lfloor\delta_1+\delta_1\ln(q^2/\delta_1)\rfloor-\lfloor\delta_2+\delta_2\ln(q^2/\delta_2)\rfloor
  \end{align*}
  where $a_2$ and $b_2$ are as in Proposition~\ref{prop:dimTricky} applied to $\delta_2$.
  
  Finally, if $q^2-q<\delta_2$, then
  \begin{equation*}
    \ell\geq q^3+q^2-2g+2-(\delta_1+\delta_2)-\sum_{s=0}^{a_1+b_1}(s+1)-\lfloor\delta_1+\delta_1\ln(q^2/\delta_1)\rfloor+\max\{a_1,0\}.
  \end{equation*}
\end{proposition}
\begin{proof}
  We use the same strategy as in the proof of Proposition~\ref{prop:codimension1}. A bound for the dimension of $\tilde{E}(\delta_1)$ can be found in Proposition~\ref{prop:dimTricky}. For $\delta_2\leq q$ the bound on $\dim\tilde{E}(\delta_2)$ comes from Proposition~\ref{prop:dimNotTricky}, and in the case $q<\delta_2\leq q^2-q$ it comes from Proposition~\ref{prop:dimTricky}. In the final case the bound follows from Proposition~\ref{prop:dimUsualHermitian}, where we note that $\delta_2\leq q^3-2g+2$ by Proposition~\ref{prop:inclusionRightMixed} and the assumption on $\delta_1$.
\end{proof}

\begin{proposition}\label{prop:codimension:3}
  Let $\delta_1,\delta_2\in H^\ast(Q)$ and $q^2-q<\delta_1<q^3-2g+2$. Further, let $\delta_2$ satisfy the conditions of Proposition~\ref{prop:inclusionRightCorner}, meaning that $\tilde{C}(\delta_2)^\perp\subseteq\tilde{E}(\delta_1)$. Denote their codimension by $\ell$.

  If $\delta_2\leq q$, we have
  \begin{equation*}
    \ell\geq q^3-g+1-\delta_1-\lfloor\delta_2+\delta_2\ln(\delta_2)\rfloor.
  \end{equation*}
  If $q<\delta_2\leq q^2-q$, then
  \begin{equation*}
    \ell\geq q^3+q^2-2g+2-(\delta_1+\delta_2)-\sum_{s=0}^{a+b}(s+1)-\lfloor\delta_2+\delta_2\ln(q^2/\delta_2)\rfloor+\max\{a,0\}
  \end{equation*}
  where $a$ and $b$ are as in Proposition~\ref{prop:dimTricky} applied to $\delta_2$.

  Finally, for $q^2-q<\delta_2$ we have
  \begin{equation*}
    \ell= q^3-\delta_1-\delta_2-2g+2.
  \end{equation*}
\end{proposition}
\begin{proof}
  Again, the the strategy is the same as in the proof of Proposition~\ref{prop:codimension1}. The exact dimension of $\tilde{E}(\delta_1)$ is given by Proposition~\ref{prop:dimUsualHermitian}. For $\delta_2\leq q$ the dimension of $\tilde{E}(\delta_2)$ can be bounded by applying Proposition~\ref{prop:dimNotTricky}, and in the case $q<\delta_2\leq q^2-q$ the bound follows by Proposition~\ref{prop:dimTricky}. For the final case we note by Proposition~\ref{prop:inclusionMiddle} that $\delta_2<q^3-q^2+q+2-(q^2-q)=q^3-4g+2$. Hence, the exact dimension of $\tilde{E}(\delta_2)$ is given by the first part of Proposition~\ref{prop:dimUsualHermitian} in this case.
\end{proof}

\begin{proposition}\label{prop:codimension4}
  Let $\delta_1,\delta_2\in H^\ast(Q)$ and $q^3-2g+2\leq\delta_1$. Further, let $\delta_2$ satisfy the conditions of Proposition~\ref{prop:inclusionRightCorner}, meaning that $\tilde{C}(\delta_2)^\perp\subseteq\tilde{E}(\delta_1)$. Denote their codimension by $\ell$. Then
  \begin{equation*}
    \ell\geq\sum_{s=0}^{a+b}(s+1)-\max\{a,0\}-\lfloor\delta_2+\delta_2\ln(\delta_2)\rfloor
  \end{equation*}
  where $q^3-\delta_1=aq+b(q+1)$ for $-q<a<q$ and $0\leq b<q$.
\end{proposition}
\begin{proof}
  The dimension of $\tilde{E}(\delta_1)$ is given by the last part of Proposition~\ref{prop:dimUsualHermitian}. To obtain a bound on the maximal value of $\delta_2$, note that the the minimal value of $q^3-\delta_1$ can be written as $q^2-q-2=q(q-2)+(q-2)$. Proposition~\ref{prop:inclusionLeftCorner} now implies $\delta_2\leq q$. Hence, $\dim\tilde{E}(\delta_2)\geq q^3-\lfloor\delta_2+\delta_2\ln(\delta_2)\rfloor$ by Proposition~\ref{prop:dimNotTricky}.
\end{proof}
The application of Theorem~\ref{thesec} or \ref{thm:css} translates Propositions~\ref{prop:codimension1}--\ref{prop:codimension4} into information on improved linear ramp secret sharing
schemes and improved asymmetric quantum codes, respectively. The details are left for the reader.

\section{Improved information on nested codes of small codimension}\label{sec:constr-with-relat}
We will now consider a second construction which in general gives nested code pairs of smaller codimension than the construction in Section \ref{sec4.5}. This construction bears some resemblance to the one given in \cite[Sec. IV]{8048519}, but in the setting of Hermitian codes.

From the definition of the codes,
$C_\mathcal{L}(D,\lambda_2Q)\subsetneq C_\mathcal{L}(D,\lambda_1Q)$ whenever $\lambda_2<\lambda_1$ and
both $\lambda_1$ and $\lambda_2$ belongs to $H^\ast(Q)$. Our second construction is captured by the following two propositions.
\begin{proposition}\label{prop:relDistConstruction2}
  Let $\lambda_1=iq+j(q+1)\in H^\ast(Q)$ where $i\leq j<q$, and define $\lambda_2=jq+i(q+1)-1$. Then $C_\mathcal{L}(D,\lambda_2Q)\subsetneq C_\mathcal{L}(D,\lambda_1Q)$ have codimension $\ell=j-i+1$, and their relative distances satisfy
  \begin{align}\label{eq:relDistConst2}
    d\big(C_\mathcal{L}(D,\lambda_1Q),C_\mathcal{L}(D,\lambda_2Q)\big) = q^3-\lambda_1 = d(C_\mathcal{L}(D,\lambda_1Q)),
  \end{align}
  and
  \begin{align}\label{eq:relDistDualConst2}
      d\big(C_\mathcal{L}(D,\lambda_2Q)^\perp,C_\mathcal{L}(D,\lambda_1Q)^\perp\big) =(i+1)(j+1) \geq d\big(C_\mathcal{L}(D,\lambda_2Q)^\perp\big).
  \end{align}
  The inequality in \eqref{eq:relDistDualConst2} is strict if and only if $i\neq 0$ and $j\neq q-1$.
\end{proposition}
\begin{proof}
  The codimension of $C_\mathcal{L}(D,\lambda_1Q)$ and $C_\mathcal{L}(D,\lambda_2Q)$ is given by the number of elements $\varepsilon$ in $H^\ast(Q)$ with $\lambda_2<\varepsilon\leq\lambda_1$. By \eqref{eq:hAstStructure} $H^\ast(Q)$ contains every integer between $\lambda_2$ and $\lambda_1$, meaning that the codimension is exactly $\lambda_1-\lambda_2=j-i+1$.

  To prove the first equalities in \eqref{eq:relDistConst2} and \eqref{eq:relDistDualConst2}, we use Proposition~\ref{prober} to obtain $\sigma(\lambda_1)=q^3-\lambda_1$ and $\mu(\lambda_2)=(i+1)(j+1)$. Applying \eqref{eqrel1} and \eqref{eqrel2} then implies that the relative distances are at least $q^3-\lambda_1$ and $(i+1)(j+1)$, respectively. That these are indeed equalities follows from the observations following Proposition \ref{prober}.
  
  In \eqref{eq:relDistConst2} the last equality follows from the last part of Proposition \ref{prober}. For \eqref{eq:relDistDualConst2} the observations in \cite[Rem.\ 4]{normtrace} imply that \eqref{eq:boundHermitianDual} is in fact an equality. Thus, the minimal distance of $d\left(C_\mathcal{L}(D,\lambda_2Q)^\perp\right)$ is given by $\mu((i+j)(q+1))=i+j+1$ if $i+j<q$ and
  \begin{align*}
    \mu\big((i-(q-1-j))q+(q-1)(q+1)\big)=q(i+j-q+2)
  \end{align*}
  otherwise. In the first case equality with $(i+1)(j+1)$ occurs if and only if $i=0$, and in the second if and only if $j=q-1$.
\end{proof}
In the above construction we only consider values of $i$ less than $q$. A similar technique can be used for $q^2-q\leq i<q^2$. We state the proposition, but omit the proof since it follows by similar arguments as above.
\begin{proposition}\label{prop:relDistConstruction2Upper}
  Let $\lambda_1=(q^2-1-i)q+(q-1-j)(q+1)\in H^\ast(Q)$ where $i\leq j <q$, and define $\lambda_2=(q^2-1-j)q+(q-1-i)(q+1)-1$. Then $C_\mathcal{L}(D,\lambda_2Q)\subsetneq C_\mathcal{L}(D,\lambda_1Q)$ have codimension $\ell=j-i+1$, and their relative distances satisfy
  \begin{align}\label{eq:relDistConst2UpperCorner}
    d\big(C_\mathcal{L}(D,\lambda_1Q),C_\mathcal{L}(D,\lambda_2Q)\big) = (i+1)(j+1)\geq d(C_\mathcal{L}(D,\lambda_1Q)),
  \end{align}
  and
  \begin{align*}
    \begin{aligned}
      d\big(C_\mathcal{L}(D,\lambda_2Q)^\perp,C_\mathcal{L}(D,\lambda_1Q)^\perp\big) = q^3-iq-j(q+1)=d\big(C_\mathcal{L}(D,\lambda_2Q)^\perp\big).
    \end{aligned}
  \end{align*}
  The inequality in \eqref{eq:relDistConst2UpperCorner} is strict if $i\neq 0$ and $j\neq q-1$.
\end{proposition}
By applying one of Theorems~\ref{thesec} and \ref{thm:css}, we can transform Propositions~\ref{prop:relDistConstruction2} and \ref{prop:relDistConstruction2Upper} into information on improved linear ramp secret sharing
schemes and improved asymmetric quantum codes, respectively. The details of this translation are left for the reader.

\section{Comparison with bounds and existing constructions}\label{sec:comp-with-exist}
Having presented two improved constructions of nested code pairs in Sections~\ref{sec4.5} and \ref{sec:constr-with-relat}, this section is devoted to the comparison between the corresponding asymmetric quantum codes and codes that already exist in the literature. The codes are also compared with the Gilbert-Varshamov bound for asymmetric quantum codes.
Moreover, we compare the corresponding secret sharing schemes with a recent lower bound on the threshold gap \cite{cryptoeprint:2018:099}.
When presenting code parameters we give the actual codimension rather
than using the estimates in Section~\ref{sec4.5} which rely on the bounds in Propositions \ref{prop:dimTricky} and \ref{prop:dimNotTricky}.

Since the codes obtained in Sections~\ref{sec4.5} and \ref{sec:constr-with-relat} are relatively long compared to the field size, the literature does not contain many immediately comparable codes.
Yet, one way to obtain such codes is by using Construction II of La~Guardia \cite[Thm. 7.1]{aq1}, which gives asymmetric quantum generalized Reed-Solomon codes. Adjusting the theorem to codes over $\mathbb{F}_{q^2}$ gives the following result.
\begin{theorem}\label{thm:laGuardia}
  Let $q$ be a prime power. There exist asymmetric quantum generalized Reed-Solomon codes with parameters
  \begin{align*}
    [[m_1m_2,m_1(2k-m_2+c),\geq d / \geq d-c]]_{q^2}
  \end{align*}
  where $1<k<m_2<2k+c\leq q^{2m_1}$ and $k=m_2-d+1$, and where $m_2,d>c+1$, $c\geq 1$, and $m_1\geq 1$ are integers.
\end{theorem}
\begin{example}\label{ex:laGuardia}
  By using different values for the parameters in Theorem \ref{thm:laGuardia}, we obtain asymmetric quantum codes of varying lengths. If the chosen parameters give a code of length less than $q^3$, we can pad each codeword with zeroes in order to obtain the correct length. Note that this does not change the relative distance of the nested codes nor of their duals.
  
  For $q=3$ Table~\ref{tab:laGuardiaQ3} lists the best code
  parameters that can be obtained in this way together with the
  comparable codes from the constructions in Sections~\ref{sec4.5} and \ref{sec:constr-with-relat}. In
  the third column, the parameter $d_z$ is maximized under the
  condition that the dimension and the distance $d_x$ are at least as
  high as in \cite{aq1}. In the fourth, the dimension is maximized,
  keeping at least the same minimal distances. As is evident, the codes
  of the present paper perform very favourably. 
  
  We further note that all presented new codes exceed the Gilbert-Varshamov bound for asymmetric quantum codes \cite[Thm. 4]{Matsumoto2017}. Additionally, we remark that nesting usual one-point Hermitian codes and using the Goppa bound does not provide asymmetric quantum codes as good as the ones in columns three and four.\goodbreak
\end{example}
The two constructions in Theorem~24 and Corollary~29 of \cite{8048519}
based on codes defined from Cartesian product point sets provide another way to obtain asymmetric quantum codes that can be compared to the ones in this paper. We summarize these constructions in the following two theorems.
\begin{theorem}\label{thm:gghr24}
  Consider integers $m\geq 2$ and $s\leq q$ where $q$ is a prime power. Given $\delta_1\in\{1,2,\ldots,s^m\}$ define $v\in\{0,1,\ldots,m-1\}$ such that $s^v\leq\delta\leq s^{v+1}$, and choose an integer $\delta_2\leq\lfloor(s-\delta_1/s^v+1)s^{m-v+1}\rfloor$. Then there exists an asymmetric quantum code with parameters
  \begin{equation*}
    [[s^m,\ell,\delta_1/\delta_2]]_q
  \end{equation*}
  where
  \begin{equation*}
    \ell\geq s^m-\sum_{t=1}^m\frac{1}{(t-1)!}\left(\delta_1\left(\ln \left(\frac{s^m}{\delta_1}\right)\right)^{t-1}+\delta_2\left(\ln\left(\frac{s^m}{\delta_2}\right)\right)^{t-1}\right).
  \end{equation*}
\end{theorem}
\begin{theorem}\label{thm:gghr29}
  Consider integers $1<s\leq q$ where $q$ is a prime power, and let $m\in\{0,1,\ldots,s-1\}$. Then for any $\ell\leq m+1$ such that $\ell$ is even if and only if $m$ is odd, there exists an asymmetric quantum code with parameters
  \begin{equation*}
    [[s^2,\ell,d_z/d_x]]_q
  \end{equation*}
  where the distances are $d_z=\frac{1}{4}\big(2s-(m-\ell+1)\big)\big(2s-(m+\ell-1)\big)$ and $d_x=\frac{1}{4}(m-\ell+3)(m+\ell+1)$. The two distances may also be interchanged.
\end{theorem}
\begin{table}[hbt]
  \centering
  \scriptsize
\begin{tabular}{*{4}{c}}
  \toprule
  \multicolumn{2}{c}{\textbf{Construction of \cite[Thm.\,7.1]{aq1}}} & \multicolumn{2}{c}{\textbf{This paper}}\\
  \cmidrule(lr){1-2}\cmidrule(lr){3-4}
  $\boldsymbol{(m_1,m_2,k,c)}$ & \textbf{Code} & \textbf{$\boldsymbol{d_z}$ maximized} & \textbf{$\boldsymbol{\ell}$ maximized}\\
  \midrule
  $(2,13,2,10)$ & $[[27, \ph2, 12 / 2]]_9$ & $[[27, \ph2, 23 / 2]]_9$  & $[[27, 12, 12 / 2]]_9$ \\
  $(2,13,3,8)$  & $[[27, \ph2, 11 / 3]]_9$ & $[[27, \ph2, 18 / 4]]_9$  & $[[27, 11, 11 / 3]]_9$          \\
  $(2,13,4,6)$  & $[[27, \ph2, 10 / 4]]_9$ & $[[27, \ph2, 18 / 4]]_9$  & $[[27, 10, 10 / 4]]_9$          \\
  $(2,13,5,4)$  & $[[27, \ph2, \ph9 / 5]]_9$  & $[[27, \ph2, 16 / 6]]_9$  & $[[27, \ph9, \ph9 / 6]]_9$            \\
  $(2,13,6,2)$  & $[[27, \ph2, \ph8 / 6]]_9$  & $[[27, \ph2, 16 / 6]]_9$  & $[[27, 10, \ph8 / 6]]_9$           \\
  $(3,9,3,4)$   & $[[27, \ph3, \ph7 / 3]]_9$  & $[[27, \ph3, 19 / 3]]_9$  & $[[27, 15, \ph7 / 3]]_9$           \\
  $(3,9,4,2)$   & $[[27, \ph3, \ph6 / 4]]_9$  & $[[27, \ph3, 17 / 4]]_9$  & $[[27, 15, \ph6 / 4]]_9$           \\
  $(2,13,3,9)$  & $[[27, \ph4, 11 / 2]]_9$ & $[[27, \ph4, 20 / 2]]_9$  & $[[27, 13, 11 / 2]]_9$          \\
  $(2,13,4,7)$  & $[[27, \ph4, 10 / 3]]_9$ & $[[27, \ph4, 18 / 3]]_9$  & $[[27, 12, 10 / 3]]_9$          \\
  $(2,13,5,5)$  & $[[27, \ph4, \ph9 / 4]]_9$  & $[[27, \ph4, 16 / 4]]_9$  & $[[27, 11, \ph9 / 4]]_9$           \\
  $(2,13,6,3)$  & $[[27, \ph4, \ph8 / 5]]_9$  & $[[27, \ph4, 14 / 6]]_9$  & $[[27, 10, \ph8 / 6]]_9$           \\
  $(2,13,7,1)$  & $[[27, \ph4, \ph7 / 6]]_9$  & $[[27, \ph4, 14 / 6]]_9$  & $[[27, 11, \ph7 / 6]]_9$           \\
  $(2,13,4,8)$  & $[[27, \ph6, 10 / 2]]_9$ & $[[27, \ph6, 18 / 2]]_9$  & $[[27, 14, 10 / 2]]_9$          \\
  $(2,13,5,6)$  & $[[27, \ph6, \ph9 / 3]]_9$  & $[[27, \ph6, 16 / 3]]_9$  & $[[27, 13, \ph9 / 3]]_9$           \\
  $(2,13,6,4)$  & $[[27, \ph6, \ph8 / 4]]_9$  & $[[27, \ph6, 14 / 4]]_9$  & $[[27, 12, \ph8 / 4]]_9$           \\
  $(2,13,7,2)$  & $[[27, \ph6, \ph7 / 5]]_9$  & $[[27, \ph6, 12 / 6]]_9$  & $[[27, 11, \ph7 / 6]]_9$           \\
  $(2,13,5,7)$  & $[[27, \ph8, \ph9 / 2]]_9$  & $[[27, \ph8, 16 / 2]]_9$  & $[[27, 15, \ph9 / 2]]_9$           \\
  $(2,13,6,5)$  & $[[27, \ph8, \ph8 / 3]]_9$  & $[[27, \ph8, 14 / 3]]_9$  & $[[27, 14, \ph8 / 3]]_9$           \\
  $(2,13,7,3)$  & $[[27, \ph8, \ph7 / 4]]_9$  & $[[27, \ph8, 12 / 4]]_9$  & $[[27, 13, \ph7 / 4]]_9$           \\
  $(2,13,8,1)$  & $[[27, \ph8, \ph6 / 5]]_9$  & $[[27, \ph8, 10 / 6]]_9$  & $[[27, 13, \ph6 / 6]]_9$           \\
  $(2,13,6,6)$  & $[[27, 10, \ph8 / 2]]_9$ & $[[27, 10, 14 / 2]]_9$ & $[[27, 16, \ph8 / 2]]_9$           \\
  $(2,13,7,4)$  & $[[27, 10, \ph7 / 3]]_9$ & $[[27, 10, 12 / 3]]_9$ & $[[27, 15, \ph7 / 3]]_9$           \\
  $(2,13,8,2)$  & $[[27, 10, \ph6 / 4]]_9$ & $[[27, 10, 10 / 4]]_9$ & $[[27, 15, \ph6 / 4]]_9$           \\
  $(2,13,7,5)$  & $[[27, 12, \ph7 / 2]]_9$ & $[[27, 12, 12 / 2]]_9$ & $[[27, 17, \ph7 / 2]]_9$           \\
  $(2,13,8,3)$  & $[[27, 12, \ph6 / 3]]_9$ & $[[27, 12, 10 / 3]]_9$ & $[[27, 17, \ph6 / 3]]_9$           \\
  $(2,13,9,1)$  & $[[27, 12, \ph5 / 4]]_9$ & $[[27, 12, \ph8 / 4]]_9$  & $[[27, 15, \ph6 / 4]]_9$           \\
  $(2,13,8,4)$  & $[[27, 14, \ph6 / 2]]_9$ & $[[27, 14, 10 / 2]]_9$ & $[[27, 19, \ph6 / 2]]_9$           \\
  $(2,13,9,2)$  & $[[27, 14, \ph5 / 3]]_9$ & $[[27, 14, \ph8 / 3]]_9$  & $[[27, 14, \ph8 / 3]]_9$           \\
  $(2,13,9,3)$  & $[[27, 16, \ph5 / 2]]_9$ & $[[27, 16, \ph8 / 2]]_9$  & $[[27, 19, \ph6 / 2]]_9$           \\
  $(2,13,10,1)$ & $[[27, 16, \ph4 / 3]]_9$ & $[[27, 17, \ph6 / 3]]_9$  & $[[27, 19, \ph4 / 3]]_9$           \\
  $(2,13,10,2)$ & $[[27, 18, \ph4 / 2]]_9$ & $[[27, 19, \ph6 / 2]]_9$  & $[[27, 21, \ph4 / 2]]_9$           \\
  $(2,13,11,1)$ & $[[27, 20, \ph3 / 2]]_9$ & $[[27, 21, \ph4 / 2]]_9$  & $[[27, 23, \ph3 / 2]]_9$\\
  \bottomrule
\end{tabular}
  \caption{Asymmetric quantum codes of length $27$ over $\mathbb{F}_9$. The first column states the parameters used in Theorem~\ref{thm:laGuardia} to obtain the codes in the second. If necessary, these have been padded with zeroes to obtain length $27$. The codes in the third and fourth columns are based on the construction in Sections~\ref{sec4.5} and \ref{sec:constr-with-relat}.}
  \label{tab:laGuardiaQ3}
\end{table}
\begin{example}
  By using different parameters in Theorems~\ref{thm:gghr24} and \ref{thm:gghr29} and padding with zeroes if
  necessary, we obtain the asymmetric quantum codes presented in
  Table~\ref{tab:gghr}. This table also shows comparable codes from
  the constructions in Sections~\ref{sec4.5} and \ref{sec:constr-with-relat} which have either $d_z$
  or $\ell$ maximized as in Example~\ref{ex:laGuardia}. From the table it is evident that the codes of the present paper perform
  very favourably.

  Again, we further note that all presented new codes exceed the Gilbert-Varshamov bound \cite[Thm. 4]{Matsumoto2017}, and that these codes cannot be constructed using information from the Goppa bound applied to nested one-point Hermitian codes.
  \begin{table}[p]
    \centering
    \scriptsize
\begin{tabular}{*{5}{c}}
  \toprule
  \multicolumn{3}{c}{\textbf{Constructions of \cite{8048519}}} & \multicolumn{2}{c}{\textbf{This paper}}\\
  \cmidrule(lr){1-3}\cmidrule(lr){4-5}
  \textbf{Type} & $\boldsymbol{(s,m)}$ & \textbf{Code} & \textbf{$\boldsymbol{d_z}$ maximized} & \textbf{$\boldsymbol{\ell}$ maximized}\\
  \midrule
  Thm.~\ref{thm:gghr29}& $(5,2)$ &$[[27, \ph1, 16 / 4]]_9$ & $[[27, \ph1, 20 / 4]]_9$  & $[[27, \ph4, 16 / 4]]_9$  \\
  Thm.~\ref{thm:gghr29}& $(5,4)$ &$[[27, \ph1, \ph9 / 9]]_9$  & $[[27, \ph1, 13 / 9]]_9$  & $[[27, \ph5, \ph9 / 9]]_9$   \\
  Thm.~\ref{thm:gghr29}& $(5,1)$ &$[[27, \ph2, 20 / 2]]_9$ & $[[27, \ph2, 23 / 2]]_9$  & $[[27, \ph4, 20 / 2]]_9$  \\
  Thm.~\ref{thm:gghr29}& $(5,3)$ &$[[27, \ph2, 12 / 6]]_9$ & $[[27, \ph2, 16 / 6]]_9$  & $[[27, \ph6, 12 / 6]]_9$  \\
  Thm.~\ref{thm:gghr29}& $(5,2)$ &$[[27, \ph3, 15 / 3]]_9$ & $[[27, \ph3, 19 / 3]]_9$  & $[[27, \ph7, 15 / 3]]_9$  \\
  Thm.~\ref{thm:gghr29}& $(5,4)$ &$[[27, \ph3, \ph8 / 8]]_9$  & $[[27, \ph3, 12 / 8]]_9$  & $[[27, \ph7, \ph8 / 8]]_9$   \\
  Thm.~\ref{thm:gghr29}& $(5,3)$ &$[[27, \ph4, 10 / 4]]_9$ & $[[27, \ph4, 16 / 4]]_9$  & $[[27, 10, 10 / 4]]_9$ \\
  Thm.~\ref{thm:gghr24}& $(5,2)$ &$[[27, \ph5, \ph7 / 1]]_9$  & $[[27, \ph5, 19 / 2]]_9$  & $[[27, 17, \ph7 / 2]]_9$  \\
  Thm.~\ref{thm:gghr29}& $(5,4)$ &$[[27, \ph5, \ph5 / 5]]_9$  & $[[27, \ph5, 13 / 6]]_9$  & $[[27, 13, \ph6 / 6]]_9$  \\
  Thm.~\ref{thm:gghr24}& $(5,2)$ &$[[27, \ph7, \ph6 / 1]]_9$  & $[[27, \ph7, 17 / 2]]_9$  & $[[27, 19, \ph6 / 2]]_9$  \\
  Thm.~\ref{thm:gghr24}& $(5,2)$ &$[[27, \ph7, \ph4 / 2]]_9$  & $[[27, \ph7, 17 / 2]]_9$  & $[[27, 21, \ph4 / 2]]_9$  \\
  Thm.~\ref{thm:gghr24}& $(5,2)$ &$[[27, \ph7, \ph3 / 3]]_9$  & $[[27, \ph7, 15 / 3]]_9$  & $[[27, 21, \ph3 / 3]]_9$  \\
  Thm.~\ref{thm:gghr24}& $(5,2)$ &$[[27, \ph8, \ph5 / 1]]_9$  & $[[27, \ph8, 16 / 2]]_9$  & $[[27, 19, \ph6 / 2]]_9$  \\
  Thm.~\ref{thm:gghr24}& $(5,2)$ &$[[27, \ph9, \ph3 / 2]]_9$  & $[[27, \ph9, 15 / 2]]_9$  & $[[27, 23, \ph3 / 2]]_9$  \\
  Thm.~\ref{thm:gghr24}& $(5,2)$ &$[[27, 10, \ph4 / 1]]_9$ & $[[27, 10, 14 / 2]]_9$ & $[[27, 21, \ph4 / 2]]_9$  \\
  Thm.~\ref{thm:gghr24}& $(5,2)$ &$[[27, 11, \ph2 / 2]]_9$ & $[[27, 11, 13 / 2]]_9$ & $[[27, 25, \ph2 / 2]]_9$  \\
  Thm.~\ref{thm:gghr24}& $(5,2)$ &$[[27, 12, \ph3 / 1]]_9$ & $[[27, 12, 12 / 2]]_9$ & $[[27, 23, \ph3 / 2]]_9$  \\
  Thm.~\ref{thm:gghr24}& $(5,2)$ &$[[27, 14, \ph2 / 1]]_9$ & $[[27, 14, 10 / 2]]_9$ & $[[27, 25, \ph2 / 2]]_9$  \\
  Thm.~\ref{thm:gghr24}& $(5,2)$ &$[[27, 17, \ph1 / 1]]_9$ & $[[27, 17, \ph7 / 2]]_9$  & $[[27, 25, \ph2 / 2]]_9$  \\
  \bottomrule
\end{tabular}
    \caption{Asymmetric quantum codes of length $27$ over $\mathbb{F}_9$. The first column indicates whether Theorem~\ref{thm:gghr24} or \ref{thm:gghr29} was used in the codes given in the third column. The second states the parameters used, except for those that can be read off directly from the code. The codes in the fourth and fifth columns are based on the construction in Sections~\ref{sec4.5} and \ref{sec:constr-with-relat}.}
    \label{tab:gghr}
  \end{table}
\end{example}
\begin{example}
  A few codes of length $8$ over $\mathbb{F}_4$ are given in \cite{Ezerman2015}. The construction in Section~\ref{sec4.5} can match -- but not improve on -- the codes  $[[8,1, 4/3]]_4$, $[[8,2,5/2]]$, $[[8,2,3/3]]_4$, $[[8,3,4/2]]_4$, and $[[8,4,3/2]]_4$. Additionally, \cite{Ezerman2015} presents a code with parameters $[[8,1,6/2]]_4$, where we can construct an $[[8,1,4/3]]_4$-code instead. All of these codes exceed the quantum Gilbert-Varshamov bound \cite[Thm.~4]{Matsumoto2017}, and the Goppa bound applied to nested one-point Hermitian codes cannot provide such parameters.

Some codes over $\mathbb{F}_9$ are presented as well. These codes do, however, have a length that is at least $36$.
\end{example}
When presenting constructions of codes, it is customary to compare it
to tables of `best known' linear codes such as
\cite{mint,tysker}. Unfortunately, similar tables do not exist for asymmetric quantum codes.
As we shall recall in a moment, however, one can still measure asymmetric quantum codes against the
usual bounds on linear codes.

Before doing so, we observe that the tables in \cite{tysker} only contains alphabets up to $\mathbb{F}_9$, whereas \cite{mint} has $\mathbb{F}_{256}$ as its largest alphabet.
As indicated by the following example, however, the latter tables are generally not as optimized as the ones by \cite{tysker}.
\begin{example}\label{ex:tableImprove}
  In some cases the improved codes $\tilde{E}(\delta)$ exceed the codes given by \cite{mint}. For instance, when considering codes over $\mathbb{F}_{16}$, the codes $\tilde{E}(12)$, $\tilde{E}(9)$, and $\tilde{E}(8)$ with parameters $[64,48,12]_{16}$, $[64,51,9]_{16}$, and $[64,53,8]_{16}$, respectively, all provide a minimal distance that is one higher than the corresponding code in the table.

  Over $\mathbb{F}_{25}$ the same is true for the codes $\tilde{E}(20)$, $\tilde{E}(16)$, and $\tilde{E}(12)$, which have parameters $[125,97,20]_{25}$, $[125,101,16]_{25}$, and $[125,106,12]_{25}$. Additionally, $\tilde{E}(15)$ has parameters $[125,103,15]_{25}$, exceeding the table distance by $2$.
\end{example}
Recall from Section~\ref{sec:constr-with-relat} that our second construction of nested code
pairs (which are code pairs of small codimension) gives impure
asymmetric quantum codes. This is already an advantage as the
error-correcting algorithms can take advantage of the impurity.
Another advantage of considering relative distances rather than only
minimal distances emerges when analysing the error-correcting ability
of asymmetric quantum codes. In order to illustrate this advantage, we
can compare nested codes from the construction in Section
\ref{sec:constr-with-relat} with pairs of best known linear codes from
the tables in \cite{tysker,mint}.
Note that the pairs of best known linear codes from such tables
generally do not result in nested code pairs; that is, they are not guaranteed to satisfy the requirement that the dual of one code is contained in the other. The comparison with tables of best known linear codes is done in the following example. Whenever the tables of \cite{mint} are considered, we will use the minimum distance of an improved algebraic geometric Goppa code from the Hermitian curve if this exceeds the table value as in Example~\ref{ex:tableImprove}.
\begin{example}
  Having fixed a code pair $C_2\subset C_1\subseteq\mathbb{F}_q^n$ of codimension $\ell$ and $d(C_1)=\delta_1$, we consider the greatest value $g(\ell,\delta)$ such that the tables of best known linear codes ensure the existence of $C,C'\subseteq\mathbb{F}_q^n$ with $\dim C-\dim C'=\ell$, $d(C)=\delta_1$, and $d(C'^\perp)\geq g(\ell,\delta_1)$. This is the same method as used in \cite{8048519}, and bears resemblance to the idea in \cite[Thm. 2]{Ezerman2015}. Using this procedure it is in no way guaranteed that $C'\subset C$. However, as shown in Table \ref{tab:ComparisonCodetables} the construction in Section \ref{sec:constr-with-relat} is in many cases on par with the best known codes, while simultaneously guaranteeing the inclusion $C_2\subset C_1$. In some cases the use of relative distances will even exceed the values obtained from the best known codes. As in the previous examples, the codes in Table~\ref{tab:ComparisonCodetables} all exceed the Gilbert-Varshamov bound for asymmetric quantum codes \cite[Thm.~4]{Matsumoto2017}.
\end{example}
\begin{table}[phb]
  \centering
  \newlength{\skipMidrule}
\setlength{\skipMidrule}{\aboverulesep+\lightrulewidth+\belowrulesep}
\footnotesize
\begin{minipage}{0.45\linewidth}
  \begin{tabular}{*{3}{c}}
    \toprule
    $\boldsymbol{(i,j)}$ & \textbf{Parameters} & $\boldsymbol{g(\ell,\delta_1)}$\\
    \midrule
    $(2,2)$ & $[[27, 1, 13 / 9]]_9$                 & $9$       \\
    $(1,1)$ & $[[27, 1, 20 / 4]]_9$                 & $4$                \\
    $(1,2)$ & $[[27, 2, 16 / 6]]_9$                 & $6$                \\
    $(0,1)$ & $[[27, 2, 23 / 2]]_9$                 & $2$                \\
    $(0,2)$ & $[[27, 3, 19 / 3]]_9$                 & $3$                \\
    \midrule
    $(3,3)$ & $[[64, 1, 37 / 16]]_{16}$             & $16$               \\
    $(2,2)$ & $\emphMath{[[64, 1, 46 / 9]]_{16}}$   & $8$                \\
    $(1,1)$ & $[[64, 1, 55 / \ph4]]_{16}$              & $4$                \\
    $(2,3)$ & $[[64, 2, 41 / 12]]_{16}$             & $12^\ast$           \\
    $(1,2)$ & $[[64, 2, 50 / \ph6]]_{16}$              & $6$                \\
    $(0,1)$ & $[[64, 2, 59 / \ph2]]_{16}$              & $2$                \\
    $(1,3)$ & $[[64, 3, 45 / \ph8]]_{16}$              & $8^\ast$            \\
    $(0,2)$ & $[[64, 3, 54 / \ph3]]_{16}$              & $3$                \\
    $(0,3)$ & $[[64, 4, 49 / \ph4]]_{16}$              & $5$                \\
    \bottomrule
  \end{tabular}
\end{minipage}
\begin{minipage}{0.45\linewidth}
  \begin{tabular}{*{3}{c}}
    \toprule
    $\boldsymbol{(i,j)}$ & \textbf{Parameters} & $\boldsymbol{g(\ell,\delta_1)}$\\
    \midrule
    $(4,4)$ & $[[125, 1, \ph81 / 25]]_{25}$            & $25$               \\ 
    $(3,3)$ & $\emphMath{[[125, 1, 92 / 16]]_{25}}$ & $14$               \\ 
    $(2,2)$ & $[[125, 1, 103 / \ph9]]_{25}$            & --                 \\ 
    $(1,1)$ & $[[125, 1, 114 / \ph4]]_{25}$            & --                 \\ 
    $(3,4)$ & $\emphMath{[[125, 2, 86 / 20]]_{25}}$ & $19$               \\ 
    $(2,3)$ & $\emphMath{[[125, 2, 97 / 12]]_{25}}$ & $10$               \\ 
    $(1,2)$ & $[[125, 2, 108 / \ph6]]_{25}$            & --                 \\ 
    $(0,1)$ & $[[125, 2, 119 / \ph2]]_{25}$            & --                 \\ 
    $(2,4)$ & $\emphMath{[[125, 3, 91 / 15]]_{25}}$ & $13$               \\ 
    $(1,3)$ & $[[125, 3, 102 / \ph8]]_{25}$            & --                 \\ 
    $(0,2)$ & $[[125, 3, 113 / \ph3]]_{25}$            & --                 \\ 
    $(1,4)$ & $[[125, 4, \ph96 / 10]]_{25}$            & $10$               \\ 
    $(0,3)$ & $[[125, 4, 107 / \ph4]]_{25}$            & --                 \\ 
    $(0,4)$ & $[[125, 5, 101 / \ph5]]_{25}$            & $6$                \\[\skipMidrule]
    \bottomrule
  \end{tabular}
\end{minipage}
  \caption{Comparing the asymmetric quantum codes from Section \ref{sec:constr-with-relat} with the best known codes. For $q=3$ \cite{tysker} is used, and for the remaining values of $q$, \cite{mint} is used. The codes marked in bold exceed $g(\ell,\delta_1)$, and the values of $g(\ell,\delta_1)$ marked with an asterisk stem from the improvements in Example \ref{ex:tableImprove}. A dash indicates that the tables do not contain enough information to determine $g(\ell,\delta_1)$.}
  \label{tab:ComparisonCodetables}
\end{table}
Turning to secret sharing schemes, \cite{cryptoeprint:2018:099}
presents a lower bound on the threshold gap $r-t$. That is, the authors bound the
smallest possible difference between the reconstruction number $r$ and
the privacy number $t$ for $q$-ary linear ramp secret sharing schemes
with $n$ shares and secrets from from
${\mathbb{F}}_q^\ell$. For linear ramp secret sharing schemes over ${\mathbb{F}}_{q^2}$, they show that
\begin{equation}\label{eq:cgrBound}
  r\geq t+\frac{(q^{2m}-1)(n+2)+(q^{2m+2}-q^{2m})(\ell-2m)}{q^{2m+2}-1}
\end{equation}
for every $m\in\{0,1,\ldots,\ell-1\}$. Of course, one should choose the $m$ that gives the best bound. Comparing the secret sharing schemes obtained in this paper with the bound \eqref{eq:cgrBound} helps to quantify how optimal the construction is. This is done in the following example, which also illustrates the advantage of using the improved codes from Section~\ref{sec4.5} and the improved information from Section~\ref{sec:constr-with-relat} rather than relying solely on the Goppa bound applied to nested one-point Hermitian codes.
\begin{figure}[hbt]
  \centering
  \begin{minipage}{0.46\linewidth}
    \centering
    \includegraphics[width=\textwidth]{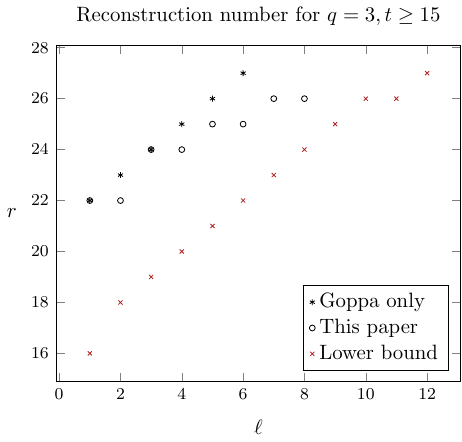}
  \end{minipage}
  \hfill%
  \begin{minipage}{0.46\linewidth}
    \centering
    \includegraphics[width=\textwidth]{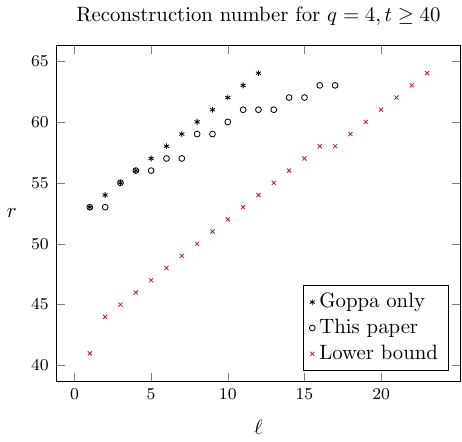}
  \end{minipage}
  \vglue 1em
  \begin{minipage}{0.46\linewidth}
    \centering
    \includegraphics[width=\textwidth]{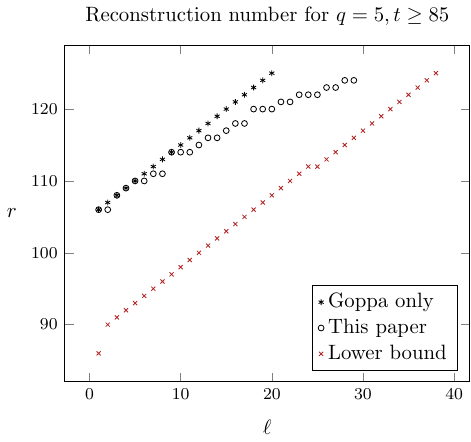}
  \end{minipage}
  \hfill%
  \begin{minipage}{0.46\linewidth}
    \centering
    \includegraphics[width=\textwidth]{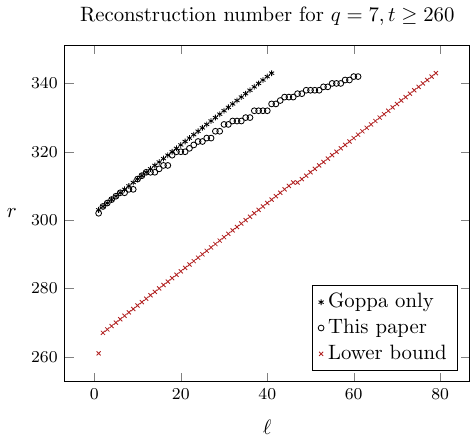}
  \end{minipage}
  \caption{The minimal achievable reconstruction number~$r$ given a desired privacy number~$t$. The plots show the constructions from Sections~\ref{sec4.5} and \ref{sec:constr-with-relat}, a construction using the Goppa bound only, and the lower bound from \eqref{eq:cgrBound}.}
  \label{fig:thresholdGap}
\end{figure}
\begin{example}
  In each of the plots in Figure \ref{fig:thresholdGap}, a desired privacy number~$t$ has been fixed.
  For each codimension the plots then show the minimal reconstruction number~$r$ achievable with the constructions from Sections~\ref{sec4.5} and \ref{sec:constr-with-relat} when privacy number at least $t$ is required. The plots also show the lower bound in \eqref{eq:cgrBound}.
  
  Recall that the codes under consideration have length $q^3$, meaning that the four corresponding secret sharing schemes support $27$, $64$, $125$, and $343$ participants, respectively.
  As the plots demonstrate, the constructions of this paper provide secret sharing schemes with parameters that could not be obtained by using nested Hermitian one-point codes and the Goppa bound alone.

  We remark that the four given examples use a relatively large privacy parameter~$t$. Yet, it is also possible to obtain improved reconstruction numbers for small values of $t$.
\end{example}

\section{Concluding remarks}\label{sec:concl}
In this paper we presented two improved constructions of nested code
pairs from the Hermitian curve, and gave a detailed analysis of their performance when applied to the concepts of secret sharing and asymmetric quantum codes. Regarding information leakage in secret sharing we studied the
reconstruction number $r$ and the privacy number $t$, which give
information on full recovery and full privacy, respectively. We note
that it is
possible to obtain information about partial information leakage by
studying relative generalized Hamming weights rather than just
relative minimum
distances~\cite{kurihara,geil2014relative}. For asymmetric quantum
codes we applied the CSS construction.
Applying the method of Steane's enlargements~\cite{MR1725135} is a future research agenda.

\section*{Acknowledgements}
The authors would like to thank Ignacio Cascudo for helpful discussions.

\appendix
\renewcommand{\thesection}{Appendix~\Alph{section}}
\section{\texorpdfstring{Additional results on $\sigma$ and $\mu$}{Additional results on sigma and mu}}\label{app:addit-results-sigma}
In this section we state a number of lemmas that are needed in Sections~\ref{sec:dimImprovedCodes} and \ref{sec:inclusion}. The lemmas all follow as corollaries to Proposition~\ref{prober}.
To aid the reader in understanding them more easily, we first give an example for reference.
\begin{example}
  In Table~\ref{tab1} we list $H^\ast(Q)$, $\sigma(H^\ast(Q))$, and
  $\mu(H^\ast(Q))$ for the Hermitian function field over
  ${\mathbb{F}}_{16}$, i.e.\ for $q=4$. Entries are ordered according to
  $(i,j)$ where $\lambda=iq+j(q+1)$. 
  \begin{table}
    \centering
    {\footnotesize
        \begin{tabular}{*{16}{c}}
          15&19&23&27&31&35&39&43&47&51&55&59&63&67&71&75\\
          10&14&18&22&26&30&34&38&42&46&50&54&58&62&66&70\\
          5&9&13&17&21&25&29&33&37&41&45&49&53&57&61&65\\
          0&4&8&12&16&20&24&28&32&36&40&44&48&52&56&60\\
          \ \\
          49 & 45 & 41 & 37 & 33 & 29 & 25 & 21 & 17 & 13 & 9  & 5  & 4  & 3  & 2 & 1\\
          54 & 50 & 46 & 42 & 38 & 34 & 30 & 26 & 22 & 18 & 14 & 10 & 8  & 6  & 4 & 2\\
          59 & 55 & 51 & 47 & 43 & 39 & 35 & 31 & 27 & 23 & 19 & 15 & 12 & 9  & 6 & 3\\ 
          64 & 60 & 56 & 52 & 48 & 44 & 40 & 36 & 32 & 28 & 24 & 20 & 16 & 12 & 8 & 4\\
          \ & \\
          4 & 8 & 12 & 16 & 20 & 24 & 28 & 32 & 36 & 40 & 44 & 48 & 52 & 56 & 60 & 64\\
          3 & 6 & 9  & 12 & 15 & 19 & 23 & 27 & 31 & 35 & 39 & 43 & 47 & 51 & 55 & 59\\
          2 & 4 & 6  & 8  & 10 & 14 & 18 & 22 & 26 & 30 & 34 & 38 & 42 & 46 & 50 & 54\\
          1 & 2 & 3  & 4  & 5  & 9  & 13 & 17 & 21 & 25 & 29 & 33 & 37 & 41 & 45 & 49
        \end{tabular}}
    \caption{Upper table: $H^\ast(Q)$. Middle table: $\sigma (H^\ast(Q))$.
        Lower table: $\mu(H^\ast(Q))$}
    \label{tab1}
  \end{table}
\end{example}
The first lemma explains when~(\ref{eqprord}) equals the Goppa bound
and when it is sharper.
\begin{lemma}\label{corstronger}
  For all $\lambda =iq+j(q+1) \in H^\ast(Q)$ it holds that $\sigma(\lambda) \geq
  n-\lambda$ where $n=q^3$. The inequality is strict if and
  only if $q^2-q \leq i < q^2$ and  $1\leq j <q\}$.
\end{lemma}
The next five lemmas give information on the relation between the values
$\sigma(iq+j(q+1))$ and $\sigma(i^\prime q+j^\prime (q+1))$ for
different constellations of $i, j, i^\prime, j^\prime$. Using the
translation from $\sigma$ to $\mu$ as given in
Proposition~\ref{prober}, this simultaneously implies relations on $\mu$. 
\begin{lemma}\label{lem1}
  For $0 < i \leq q^2-q-1$ and $0 \leq j < q-1$ it holds that
  $\sigma(iq+j(q+1))=\sigma ((i-1)q+(j+1)(q+1))+1$. Furthermore, for $0 \leq i
  \leq q^2-q-1$ it holds that $\sigma(iq+(q-1)(q+1))=\sigma((i+q)q)+1$. 
\end{lemma}
\begin{lemma}\label{lem:sequence}
  The sequence
  \begin{align*}
    \big(\sigma(0\cdot &q),  \sigma(q), \ldots , \sigma((q^2-1)q),
                         \sigma((q^2-1)q+(q+1)),\\
                       &\sigma((q^2-1)q+2(q+1)), \ldots, \sigma((q^2-1)q+(q-1)(q+1))\big)
  \end{align*}
  is strictly decreasing.
\end{lemma}
\begin{lemma}\label{lem:symmetricCorner}
  We have
  $\sigma((q^2-q+s)q+t(q+1))=\sigma((q^2-q+t)q+s(q+1))$ for $0 \leq s, t<q-1$.
\end{lemma}
\begin{lemma}\label{lem:improvedCorner}
  Given $q^2-q \leq i \leq q^2-1$ then for non-negative $s$ such that
  $q^2-q\leq i-s$ we have $\sigma((i-s)q+s(q+1)) \geq \sigma
  (iq)$. Similarly, given $0 \leq j \leq q^2-1$ then for non-negative
  $s$ such that $j+s \leq q-1$ we have $\sigma((q^2-1-s)q+(j+s)(q+1)) \geq \sigma((q^2-1)q+j(q+1))$.
\end{lemma}
\begin{lemma}\label{lem:smallSigma}
  If $\sigma(iq+j(q+1))\leq q$ then $q^2-q \leq i < q^2$.
\end{lemma}
Finally, we present a lemma on the relation between $\sigma(\lambda)$ and
$\mu(\lambda)$ for $\lambda$ belonging to a certain window. 
\begin{lemma}\label{lem6}
  For $\lambda=iq+j(q+1) \in H^\ast(Q)$ with $q \leq i <q^2-q$ and $j$ arbitrary; or $q^2-q
  < i \leq q^2-1$ and $j=0$; or $0 \leq i < q$ and $j=q-1$, we
  have $\mu(\lambda)+\sigma(\lambda)=q^3-(q^2-q-1)$.
\end{lemma}

\section*{References}\renewcommand{\chapter}[2]{\footnotesize}
\bibliography{bibfile}
\bibliographystyle{plainurl}

\end{document}